\documentclass[journal]{new-aiaa}
\usepackage[utf8]{inputenc}
\usepackage{textcomp}

\usepackage{graphicx}
\usepackage{float,subcaption}

\usepackage{amsmath,amssymb,latexsym,amsthm,amsfonts,mathtools}

\usepackage{siunitx}
\usepackage{booktabs}

\theoremstyle{plain}
\newtheorem{theorem}{Theorem}
\newtheorem{lemma}{Lemma}

\theoremstyle{remark}
\newtheorem{remark}{Remark}
\usepackage{hyperref}
\usepackage[nameinlink]{cleveref}

\usepackage{graphicx}
\usepackage{amsmath}
\usepackage[version=4]{mhchem}
\usepackage{siunitx}
\usepackage{longtable,tabularx}
\newtheorem{definition}{Definition}
\newtheorem{assumption}{Assumption}

\newcommand{\tr}[1]{\textcolor{red}{#1}}

\title{Engagement-Zone-Aware Input-Constrained Guidance for Safe Target Interception in Contested Environments}

\author{Praveen Kumar Ranjan \footnote{Postdoctoral Fellow,  email: praveen.ranjan@my.utsa.edu, Member AIAA}.}
\affil{Unmanned Systems Lab, Department of Electrical Engineering,\\ The University of Texas at San Antonio, San Antonio, TX, 78249}
\author{Abhinav Sinha \footnote{Assistant Professor, email: abhinav.sinha@uc.edu $^\ddag$ Corresponding author, Senior Member AIAA.} }
\affil{Guidance, Autonomy, Learning, and Control for Intelligent Systems (GALACxIS) Lab, \\ Department of Aerospace Engineering and Engineering Mechanics, University of Cincinnati, OH 45221, USA}
\author{Yongcan Cao \footnote{Professor, email: yongcan.cao@utsa.edu, Senior Member AIAA.}}
\affil{Unmanned Systems Lab, Department of Electrical Engineering,\\ The University of Texas at San Antonio, San Antonio, TX, 78249}

\begin{document}

\maketitle

\begin{abstract}
This paper addresses target interception in contested, GPS-denied environments in the presence of multiple moving defenders whose interception capability is limited by finite engagement ranges. Conventional methods typically impose conservative stand-off constraints based solely on maximum engagement distance and neglect the actuator limitations of the interceptor. Instead, we formulate safety constraints using defender-induced engagement zones (EZs), defined as regions of the attacker’s state space from which capture becomes inevitable under the defenders’ current engagement capability. To explicitly account for actuator limits, the nonholonomic vehicle model is augmented with symmetric input saturation dynamics. A time-varying safe-set tightening parameter is introduced to compensate for transient constraint violations induced by actuator dynamics. To ensure scalable safety enforcement in multi-defender scenarios, a smooth aggregate safety function is constructed using a log-sum-exp (soft-minimum) operator that combines individual threat measures associated with each defender’s engagement capability. A smooth switching guidance strategy is then developed to coordinate interception and safety objectives. The attacker pursues the target when sufficiently distant from threat boundaries and progressively activates evasive motion as the EZ boundaries are approached. The resulting controller relies only on relative measurements and does not require knowledge of defender control inputs, thus facilitating a fully distributed and scalable implementation. Lyapunov-based analysis provides sufficient conditions guaranteeing target interception, practical safety with respect to all defender engagement zones, and satisfaction of actuator bounds. An input-constrained guidance law based on conservative stand-off distance is also developed to quantify the conservatism of maximum-range-based safety formulations. Simulation studies with stationary and maneuvering defenders demonstrate that the proposed EZ-aware formulation yields shorter interception paths and reduced interception time compared with conventional maximum-range-based safety methods while maintaining safety throughout the engagement.
\end{abstract}

\section{Introduction}
The growing need for autonomous systems operating in contested environments necessitates guidance strategies that ensure both survivability and mission accomplishment \cite{doi:10.2514/1.37030,doi:10.2514/1.G007057,doi:10.2514/1.G003157,10634571,10660561,10839025}. Applications such as target interception, air defense, and precision strike missions place requirements on guidance strategies in addition to ensuring negligible miss distance. Such requirements have been addressed in the form of terminal constraints, e.g., impact time (see \cite{9000526,doi:10.2514/1.G005367,doi:10.2514/1.G005180} and references therein). However, in many settings, the pursuer must often breach defended regions to engage with the target and function under strict sensing and maneuvering limitations. A fundamental abstraction for analyzing these interactions is the Target–Attacker–Defender (TAD) engagement, where a pursuer (or attacker) attempts to reach a target while one or more defenders employ interception strategies to neutralize the attacker \cite{sinha2022three,9274339}.

While defender–target cooperation \cite{7171913,9274339,Casbeer2018,sinha2022three,doi:10.2514/6.2010-7876,doi:10.2514/1.58566} has been the primary focus in most prior studies of TAD games, only limited attention has been given to strategies from the attacker’s perspective. In \cite{doi:10.2514/6.2025-1902}, the authors presented a pursuit strategy under an adversarial environment by leveraging reinforcement learning and data-driven methodologies. The three-agent pursuit-evasion dynamics is further investigated in \cite{doi:10.2514/1.61832}, where differential game formulations are employed to derive sufficient conditions for the attacker to strike the target while avoiding interception by the defenders. The authors in \cite{SUN20192337} designed guidance laws to steer the attacker to a defender-safe boundary and then maintain a critical miss distance, establishing attacker-win conditions under linearized game dynamics without requiring knowledge of the target/defender's control efforts. In \cite{doi:10.2514/1.51611}, a linear quadratic differential game was used to derive cooperative pursuit–evasion strategies, including the homing interceptor’s optimal pursuit and evasion policy. The authors in \cite{QI20171958} utilized a three-player bounded-control game to design an attacker strategy that guarantees a miss distance from the defender by avoiding the infeasible zero effort miss region and establishes sufficient conditions for the attacker's success. 

It is worth mentioning that most of the safety constraints in interception and pursuit-evasion problems are typically enforced via feasibility constraints, e.g., minimum distance from obstacle or threat region, and ad-hoc actuator limits during implementation. As a result, such considerations do not explicitly characterize the effective input bounds for safety preservation. For example, the authors in \cite{doi:10.2514/1.G003223} developed an intercept-angle target interception law for a multiple-obstacle environment by enforcing safety constraints through minimum-distance constraints within an optimal guidance formulation. In \cite{8263864}, a multiple TAD differential game was studied, where target safety is ensured through optimal defender-attacker pairing strategies and feedback control laws derived from double-integrator dynamics. Other works have addressed safety via input constraints. In \cite{RUSNAK20119349}, the control effort of the players is penalized in a TAD game through quadratic costs, resulting in the bounds appearing as part of the optimization objective rather than as explicit safety constraints. In \cite{11288066}, guidance law was developed for impact-time interception that explicitly accounts for seeker field of view constraints and actuator bounds by incorporating an input affine saturation model. In \cite{doi:10.2514/1.47276}, the authors developed a multi-model adaptive estimator-guidance framework utilizing model-dependent gains to account for actuator saturation and target maneuver uncertainty. The authors in \cite{ranjan2025incp} developed an input-constrained guidance law for an attacker in a multiple static defender threat environment using a discontinuous switching function. A time-constrained target-intercept guidance strategy was developed in \cite{doi:10.2514/6.2026-0121} to lead an attacker toward a target while avoiding multiple obstacles of varying sizes.

Most of the above-mentioned works require complete knowledge of the defender's strategy to guarantee the attackers' escape and employ simplified vehicle models (e.g., linearized dynamics). This assumption may be impractical in realistic scenarios, especially those involving multiple defenders with range and maneuverability constraints. A more practical approach is to design the intercept guidance strategy using geometric threat sets rather than exact defender threat predictions. An engagement zone (EZ) defines the set of attacker-relative states from which a defender can guarantee interception, assuming the attacker does not change course. Several studies have explored modeling engagement or capture zones under different structural constraints (e.g., static obstacles \cite{OYLER20161}, constrained environments \cite{ZHOU201664}, and visibility \cite{IBRAGIMOV1998187}). An alternate relevant formulation involves modeling the defenders as range-limited, where each defender can travel only up to a maximum range, as considered in \cite{10365808, doi:10.2514/1.I011394,doi:10.2514/1.I011593}. These modeling approaches offer ways to encode threat geometry, which helps develop safety-aware guidance strategies in constrained environments. 

Unlike static safety margins or fixed threat envelopes, EZs are dynamic, geometry-dependent, and can be analytically characterized. This formulation incorporates key factors such as velocity ratios, turning constraints, and capture radii. Therefore, incorporating EZ information into guidance design enables more accurate identification of the safe maneuvering region to reduce conservatism and improve interception efficiency in a contested environment. Moreover, the analytical characterization of the EZs allows for the safety to be evaluated directly from engagement geometry, circumventing the need for repeated optimization steps typically required in constrained optimal control or model predictive control approaches. Additionally, practical interceptor systems operate under input constraints, and neglecting these constraints may result in infeasible guidance commands that compromise safety during the engagement. However, most existing approaches rely on computationally expensive, optimization-based methods that use global information and do not explicitly guarantee safety or input constraint satisfaction. Despite substantial progress in characterizing EZs and capture regions, most methods rely on online optimization that treats EZs and input bounds as constraints, leading to significant computational overhead and limiting their applicability in fast, decentralized, and resource-constrained environments, especially when multiple defenders are present. Motivated by these limitations, we develop an engagement-zone-aware guidance framework for the attacker that explicitly incorporates input constraints while ensuring safety with respect to defender-induced threats. To the best of our knowledge, this is the first work that systematically integrates analytic EZ characterization with input-constrained feedback design for safe target interception in a contested environment. The main contributions of this work are summarized below.

 First, we develop a nonlinear switching guidance strategy using only relative measurements that enables the attacker to intercept a target in the presence of multiple moving defender threats while ensuring safety and satisfying input constraints. Unlike optimization-based or differential game approaches, the proposed strategy relies on analytical safety characterization and provides the guidance law in a feedback form to enable computationally efficient real-time implementation.

{Second, the safety guarantees in our proposed design are twofold-- (i)  safety with respect to defender-induced threats, and (ii) safety ensuring permissible bounded control input while maintaining safety or stability. Our work aggregates individual defender EZ-based threats via a smooth log-sum-exp approximation to account for multiple defenders. Therefore, our method reduces the conservatism associated with prior methods that rely on fixed stand-off distances to enable efficient interception in contested environments.}

{Third, the proposed design provides safety guarantees by augmenting the engagement kinematics with a smooth symmetric input saturation model and proposing a safe-set tightening parameter that shrinks the safe set to enable the vehicle to preemptively apply deceleration when heading towards EZs. Unlike prior methods that treat actuator limits only as feasibility constraints or ignore their impact on safety, the proposed approach explicitly integrates input constraints into the safety-critical guidance design. Moreover, we provide analytical guarantees on safety and interception performance under bounded control inputs.}

{Fourth, we also develop a guidance strategy based on conservative stand-off distance constraints and contrast it with the engagement zone formulation for further insights. In both scenarios, we establish theoretical guarantees showing that the proposed strategies preserve safety with respect to all defender engagement zones while satisfying input constraints.}

\section{Preliminaries and Problem Formulation}
This section formulates the attacker-defender-target engagement kinematics in relative coordinates, incorporating bounded-rate dynamics for the control input. We then characterize the defender-induced EZs and formulate the control objectives for the attacker to intercept the target while maintaining safety from the defenders.

\subsection{Vehicles' Relative Kinematics Model}
Consider a multi-agent engagement scenario consisting of a mobile attacker $A$, $n$ moving defenders denoted by the set $\mathcal{D} = \{D_1, D_2, \ldots, D_n\}$, and a stationary target $T$, as illustrated in \Cref{fig:enggeom}. Each defender is equipped with an interception capability that can only reach a finite distance from its instantaneous position. This maximum engagement distance is represented by the dotted black circles centered at each defender in \Cref{fig:enggeom}.
\begin{figure}[h!]
    \centering
    \includegraphics[width=.5\linewidth]{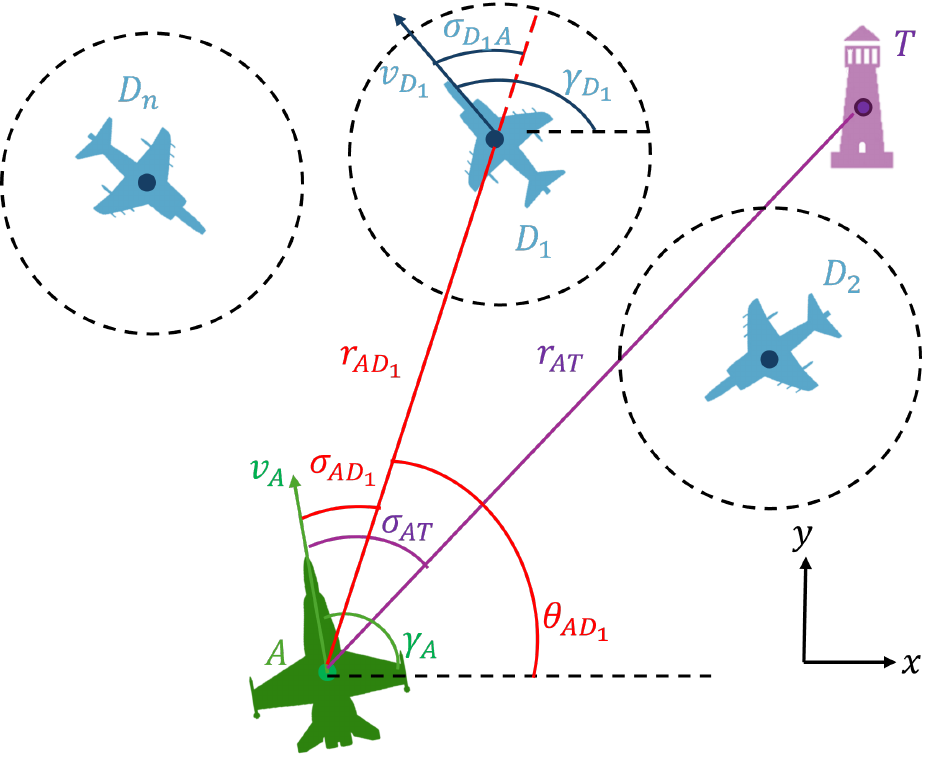}
    \caption{Attacker-Target-Defender engagement geometry.}
    \label{fig:enggeom}
\end{figure}

The attacker and the defenders are modeled as nonholonomic vehicles moving at constant speed. Their kinematics in the inertial frame are described by
\begin{align}
\dot{x}_i = v_i \cos \gamma_i, \quad \dot{y}_i = v_i \sin \gamma_i, \quad \dot{\gamma}_i = \frac{a_i}{v_i},
\label{eqn:kine_inertial}
\end{align}
where the subscript $i$ denotes the attacker ($i=A$) or the $j$\textsuperscript{th} defender ($i=D_j$) with $j \in\mathcal{D}$, $[x_i, y_i]^\top \in \mathbb{R}^2$ denotes the vehicle position, $v_i > 0$  denotes the vehicle speed, $\gamma_i \in(0, 2\pi]$ denotes the vehicle heading angle and $a_i$ denotes the lateral acceleration that is the only control input to the vehicle model. This vehicle model captures the motion of turn-constrained platforms such as fixed-wing aircraft, missiles, and underwater vehicles that maneuver through lateral forces (e.g., lift or side force). To explicitly account for bounded control inputs, we augment the attacker dynamics with a smooth symmetric input-saturation model \cite{kumar2025provably},
\begin{align}
    \dot{a}_A = \left[1-\left(\frac{a_A}{a_{\max}}\right)^n\right]a_A^c - p_1 a_A, \label{eqn:sat_model}
\end{align}
where $n=2$, $p_1\in\mathbb{R}_{> 0}$ denote constants, and $a_A^c$ denotes the commanded lateral acceleration, which is the pseudo control input for the augmented system incorporating the above saturation model. 
\begin{remark}\label{rmk:inputsat}
    From \eqref{eqn:sat_model}, it follows that when $|a_A|\to a_{\max}$, $\dot{a}_A\to -p_1a_A,$ and when $|a_A|\to -a_{\max}$, $\dot{a}_A\to p_1a_A,$ which implies that $a_A$ will decrease if $a_A\to a_{\max}$ and $a_A$ will increase if $a_A\to -a_{\max}$. Therefore, the above saturation model ensures that $a_A$ remains bounded as $|a_A|<a_{\max}$. Thus, we augment the saturation model in \eqref{eqn:sat_model} with the kinematics \eqref{eqn:rel_dyn_1}-\eqref{eqn:rel_dyn_2} and will design $a_A^c$, which will automatically ensure that $a_A$ remains within the permissible limits. 
\end{remark}
To facilitate guidance design with only relative information, we transform the given motion model into relative polar coordinates. This transformation helps us to reduce the dimensionality of the problem and the complexity of the control, while offering intuitive geometric information on the attacker's motion relative to the target and the defenders. Based on the inertial positions, the relative distance $r_{Aj} \in \mathbb{R}_{\geq0}$ and the line-of-sight (LOS) angle $\theta_{Aj} \in[0,2\pi)$ between the attacker and the $j$\textsuperscript{th} agent (defenders or target) are defined as,
 \begin{align}
     r_{Aj} = \sqrt{\left(x_{j}-x_{A}\right)^2+\left(y_{j}-Y_{A}\right)^2}, ~\theta_{Aj} = \tan^{-1}\left(\frac{y_{j}-y_{A}}{x_{j}-x_{A}}\right), \label{eqn:rel_dis}
 \end{align}
 where the subscript $j \in \mathcal{D} \cup \{T\}$ and the LOS is measured from the attacker to the $j$\textsuperscript{th} vehicle. Additionally, we define the attacker's lead angle as the angle subtended by the attacker's velocity to the respective LOS to defenders or target, and is given by,
 \begin{align}
     \sigma_{Aj} = \gamma_A-\theta_{Aj}, ~\forall \; j\in \mathcal{D} \cup \{T\}, \label{eqn:lead}
 \end{align}
 such that $\sigma_{Aj} \in (-\pi, \pi]$. Similarly, we define the defender's lead angle with respect to the attacker as,
 \begin{align}
     \sigma_{j A} = \gamma_{j} - \theta_{Aj}, ~\forall j \in \mathcal{D}, \label{eqn:rev_lead}
 \end{align}
 representing the angle subtended by the defender’s velocity vector with respect to the LOS angle from the attacker to the defender. 
 
 Based on the above-defined relative variables, we can express the relative motion kinematics between the attacker and the $j$ \textsuperscript{th} agent as,
\begin{align}
    \dot{r}_{Aj} &=v_j\cos\sigma_{jA}-v_A\cos\sigma_{Aj}\label{eqn:rel_dyn_1}\\
    r_{Aj}\dot{\theta}_{Aj} &= v_j\sin\sigma_{jA}-v_A\sin\sigma_{Aj}, \; \forall \; j \in \mathcal{D}\cup\{ T\}, \label{eqn:rel_dyn_2}
    \end{align}
using \eqref{eqn:kine_inertial}, \eqref{eqn:rel_dis}, \eqref{eqn:lead} and \eqref{eqn:rev_lead}. The relative formulation simplifies the representation of the attacker’s motion and plays a crucial role in characterizing threat regions posed by the defenders.

\subsection{Defender-induced Engagement Zones (EZs)}
In this work, each defender is modeled as a mobile platform, such as an aircraft equipped with interception capabilities (e.g., turret-mounted systems or missiles). The instantaneous position of the defender is treated as the origin from which this capability can be deployed. The objective is to design a guidance strategy that enables the attacker to remain safe from defenders by identifying and avoiding regions of the state space where interception becomes unavoidable under the current control policy. These regions are referred to as EZs.
\begin{definition}[Engagement Zone (EZ)]\label{def:ez}
The EZ of a defender is defined as the subset of the attacker's state space from which the defender’s engagement capability guarantees interception, assuming the attacker continues with its current heading and control strategy.
\end{definition}
\begin{remark}
    Therefore, EZ characterizes the region in the attacker's configuration space where it is at risk of being intercepted by the defenders without altering its trajectory. In the relative polar coordinate frame centered at the defender's point of origin, the \textit{engagement boundary} represents the outermost surface of the EZ, beyond which the attacker remains safe but crossing which guarantees interception if it maintains its current trajectory. 
\end{remark}
For the fast defenders $D_i$, with speed ratio $\mu_i=\dfrac{v_A}{v_{D_i}} \in(0,1)$, the boundary of its EZ is analytically characterized as
\begin{align}
\rho_i (\sigma_{Ai}) = \mu_i R_i \left[ \cos \sigma_{Ai} + \sqrt{\cos^2 \sigma_{Ai} - 1 + \frac{(R_i + c_i)^2}{\mu_i^2 R_i^2}} \right], \;\forall\; i \in \mathcal{D},\label{eqn:eng_bound}
\end{align}
where $\rho_i$ denotes the critical radial distance from a defender's point of origin, $R_i$ denotes the maximum engagement range, and $c_i$ denotes the capture radius. 
\begin{figure}[h!]
	\centering
	\begin{subfigure}[t]{.49\linewidth}
    \centering
    \includegraphics[width=\linewidth]{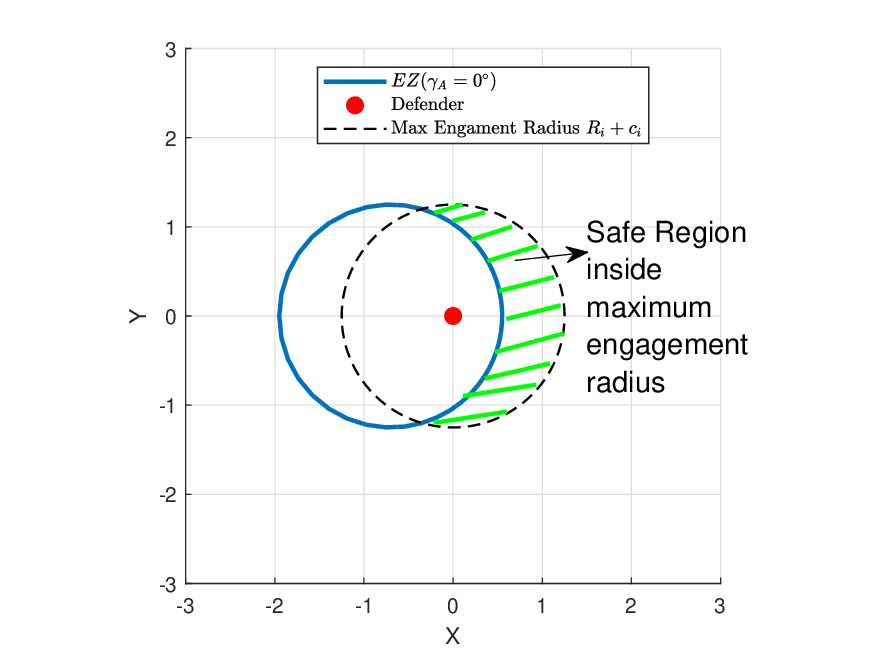}
    \caption{Safe region with EZ considerations.}
    \label{fig:sing_safe}
    \end{subfigure}
	\begin{subfigure}[t]{0.49\linewidth}
		\centering		\includegraphics[width=\linewidth]{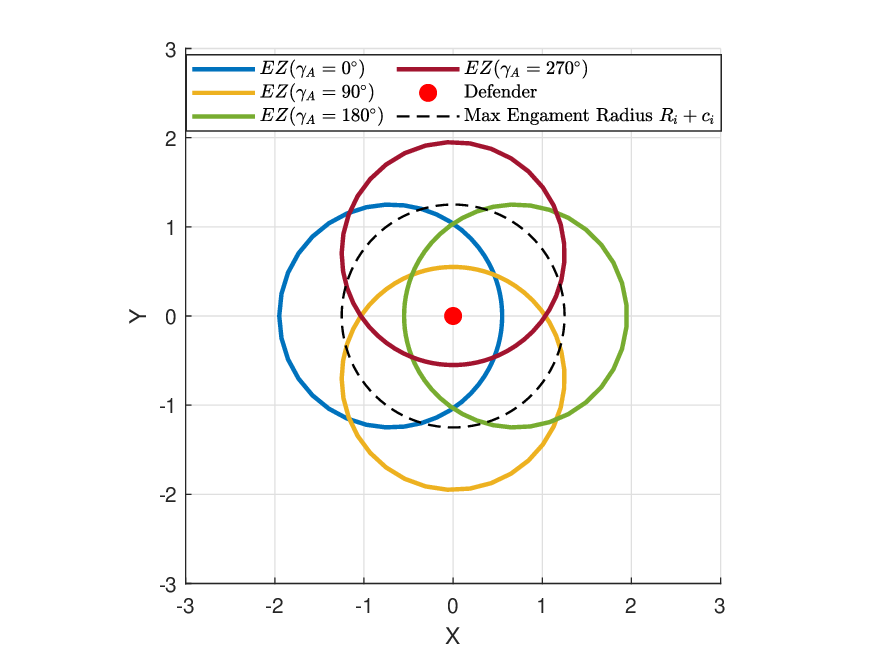}
		\caption{Comparison of EZ for different Attacker headings.}
		\label{fig:sing_safe_1}
	\end{subfigure}	
	\caption{Illustration of EZ for single defender ($[x_D, y_D]^\top=[0 \;0]^\top$ with $R_D=1$ m and $c_D=0.25$m).}
	\label{fig:ez_sing_comp}
\end{figure}
\Cref{fig:sing_safe} illustrates the engagement boundary for a single defender when the attacker maintains a constant heading $\gamma_A=0^\circ$, where it is observed that the EZ shifts left relative to the defender's maximum engagement radius, which creates a region within the nominal range where the attacker can remain safe. This observation highlights that using the maximum engagement radius as the engagement boundary is overly conservative, as it unnecessarily restricts feasible attacker trajectories and can lead to longer paths and increased interception time. More generally, the position of the engagement zone depends on the lead angle between the attacker and the defender. As the attacker’s strategy changes, the engagement boundary shifts accordingly, and causes the safe region within the maximum engagement radius to move in different directions, as illustrated in \Cref{fig:sing_safe_1}.
\begin{figure}[h!]
    \centering
    \includegraphics[width=0.5\linewidth]{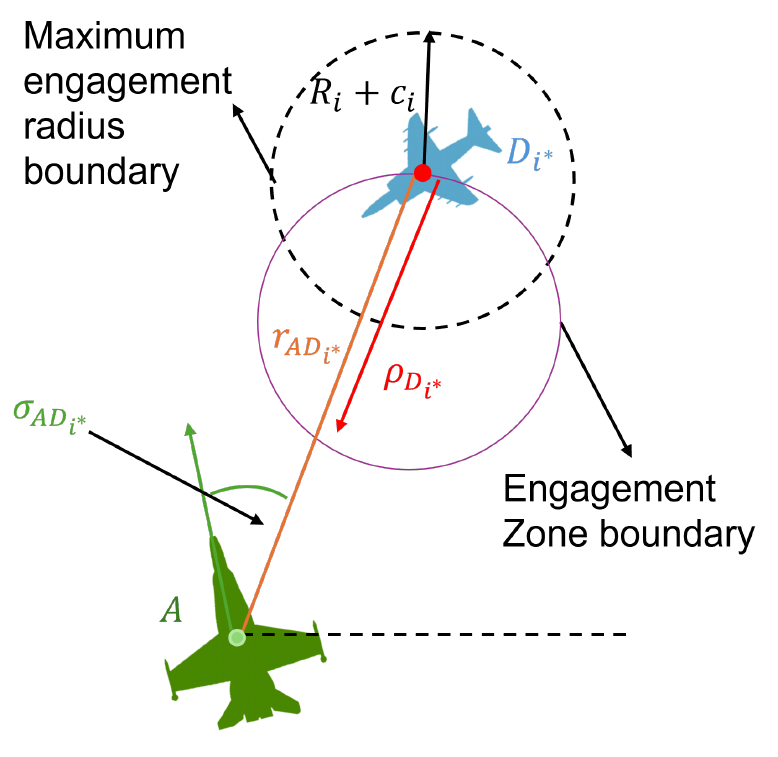}
    \caption{Engagement boundary relative to the maximum range.}
    \label{fig:eng_reg}
\end{figure}

Based on the above properties, we can define the region of state space inside each defender's maximum engagement range as
\begin{align}
\mathcal{E}^{MR} = \left\{ \left(r_{Ai},\theta_{Ai},\sigma_{Ai}\right) \in \mathbb{R}_{\geq 0} \times [0,2\pi)\times (-\pi, \pi] \, \mid r_{Ai}\leq R_i+c_i, i \in \mathcal{D} \right\}.\label{eq:emr}
\end{align}
The expression in \eqref{eq:emr} implies that remaining outside the set $\mathcal{E}^{MR}$ enables the attacker to conservatively maintain safety by remaining outside the maximum engagement radius of all the defenders. In contrast, based on \Cref{def:ez}, the EZ of a defender $D_i$ with respect to the attacker is defined as
\begin{align}  
\mathcal{E} := \left\{ \left(r_{Ai},\theta_{Ai},\sigma_{Ai}\right) \in \mathbb{R}_{\geq 0} \times [0,2\pi)\times (-\pi, \pi] \,\middle|\, r_{Ai} \leq \rho_i(\sigma_{AD_i}), i \in \mathcal{D} \right\}.
\end{align}
Crossing this EZ boundary places the attacker within the defender’s guaranteed capture region. Therefore, the engagement boundary accounts for relative engagement geometry, including velocity ratio and maneuvering capabilities, providing an accurate characterization of the defender's capture region. Additionally, employing EZ-based capture region allows the attacker to exploit maneuvering regions that would otherwise be deemed unsafe under conservative formulations. To better clarify the different regions of state-space, \Cref{fig:eng_reg} illustrates different regions around the defender for the current relative side bearing angle of the attacker (according to its current strategy). The dotted black circle of radius $R_i+c_i$ represents the maximum defender range that represents the conservative stand-off safety constraint. The purple colored circle represents the EZ boundary, inside which the attacker's capture is guaranteed. It is important to note that the sets $\mathcal{E}$ and $\mathcal{E}^{MR}$ generally overlap but do not coincide. While $\mathcal{E}^{MR}$ remains fixed with respect to the defender, the set $\mathcal{E}$ is subject to change as the attacker maneuvers around the defender.

\subsection{Design of the Control Objectives}
The objective of this paper is to design a nonlinear guidance law that enables the attacker to intercept a stationary target while remaining safe from threats posed by multiple moving defenders and respecting bounds on the attacker’s lateral acceleration. Achieving this objective requires the attacker to satisfy three coupled requirements.

First, the guidance law must ensure target interception by asymptotically reducing the relative distance between the attacker and the target to zero, i.e.,
\begin{align}
\lim_{t \to \infty} r_{AT}(t) = 0.
\end{align}
Second, the attacker must guarantee EZ avoidance by ensuring that its trajectory remains outside the engagement zone of every defender at all times. Formally, this requires
\begin{align}    
(r_{Ai},\theta_{Ai}, \sigma_{Ai}) \notin \mathcal{E}_i , ~\forall~  t\geq0\end{align}
which is equivalently expressed as maintaining $r_{Ai} > \rho_i$ for each defender. Additionally, we design the attacker's strategy to guarantee conservative safety by ensuring the attacker remains outside the defender's maximum engagement range. Such a design forms our baseline approach for comparison and is formally given as
\begin{align}
    (r_{Ai}, \theta_{Ai}, \sigma_{A_i}) \notin \mathcal{E}^{MR} , ~\forall t\geq0 \label{eqn:cns_saf_cond}
\end{align}
that will enable the attacker to maintain $r_{Ai}>R_i+c_i,\;\forall \; i \in \mathcal{D}$.

Third, the guidance strategy must satisfy actuator constraints by explicitly enforcing bounds on the attacker’s lateral acceleration during design. This requirement is expressed as
\begin{align}
\left\vert a_A(t) \right\vert \le a_A^{\max},
~ \forall t \ge 0,
\end{align}
where $a_A^{\max} > 0$ denotes the prescribed symmetric bound on the attacker’s lateral acceleration.

\section{Engagement-Zone-Aware Input-Constrained Intercept Guidance Design} \label{sec:our_app}
In this section, we develop the nonlinear guidance law for the attacker to intercept the target while preventing neutralization from the defender-induced EZs, while also accounting for physical limits on the attacker's lateral acceleration. We first formulate the safety conditions that ensure the attacker remains outside the regions reachable by the defenders’ maximum engagement range, by modeling unsafe regions as time-varying EZ boundaries rather than fixed standoff radii. Using these safety conditions, we construct an aggregate safety measure that combines the effects of all defender-induced EZs and includes a tunable parameter that controls how strongly multiple EZs contribute to the overall safety measure. Finally, we design a smooth switching guidance law for the attacker in the presence of multiple defenders that guarantees both the target interception and avoidance of defender EZs under the attacker's input constraints. 

To incorporate the defender-induced EZs in the control design, we define a safety function for the attacker with respect to each EZ as
\begin{align}
    b_i= r_{AD_i}-\rho_i, \; i \in \mathcal{D},
\end{align}
which depends on the attacker's position and velocity, yielding a characterization of the safe state-space based on the attacker's current strategy (heading angle). Note that $b_i>0$  represents the scenarios when the attacker is outside the i\textsuperscript{th} defender's EZ, that is, $r_{AD_i}>\rho_i$. On the other hand, $b_i\leq 0$ indicates scenarios when the attacker is present inside the i\textsuperscript{th} defender's EZ, implying that interception by the defender’s deployed capability is guaranteed. Therefore, maintaining safety requires $b_i>0$ for all defender-induced EZs.  Accordingly, the safe set for the attacker can be defined as
\begin{align}
    \mathcal{C} \coloneqq \left\{\, (x_A, y_A, \gamma_A) \in \mathbb{R}^2 \times [0,2\pi) \;\big|\; b_i > 0 \; \forall \;i\in\mathcal{D} \,\right\}.
    \label{eq:safe_set}
\end{align}
To ensure the safety of the attacker (forward invariance of set $\mathcal{C}$), it is necessary to characterize how each safety function evolves along the system trajectories. Therefore, differentiating $b_i$ with respect to time and using the relative kinematics model, we obtain the dynamics of the safety function as
\begin{align}
\dot{b}_i&=v_i\cos\sigma_{iA}-v_A\cos\sigma_{Ai} +\nabla\rho_i\frac{v_i\sin\sigma_{iA}}{r_{Ai}}-\nabla\rho_i\frac{v_A\sin\sigma_{Ai}}{r_{Ai}}-\nabla\rho_i\frac{a_A}{v_A} \label{eqn:bdot},
\end{align}
where $\nabla\rho_i$ denotes the gradient of engagement boundary of the $i$\textsuperscript{th} defender, computed using \eqref{eqn:eng_bound}, as
\begin{align}
\nabla \rho_i = -\mu_i R_i\sin\sigma_{AD_i} \left[1+\frac{\cos\sigma_{AD_i}}{\sqrt{\cos^2 \sigma_{A D_i} - 1 + \dfrac{(R_i + c_i)^2}{\mu_i^2 R_i^2}}}\right].    
\end{align}
From the expression of $\dot{b}_i$, it can be observed that $a_A$ has relative degree one with respect to the safety constraint $b_i$. 

To account for multiple EZs, we aggregate the safety function for all the EZs using the log-sum-exp (soft-minimum) operator, 
\begin{align}
    h = -\beta\log\left(\sum_{i \in\mathcal{D}}e^{-b_i/\beta}\right),\label{eq:h}
\end{align}
where $\beta\in\mathbb{R}_{>0}$ is a user-defined constant that determines how closely $h$ approximates the minimum values of individual safety functions.
\begin{remark}\label{rem:agg_safety_cond}
    The safety aggregation function in \eqref{eq:h} satisfies
    \begin{align}
        \min_{i\in\mathcal{D}}b_i-\beta\log n \leq h \leq \min_{i\in\mathcal{D}}b_i,  \label{eqn:agg_safety_bounds}
    \end{align}
    where $n$ denotes the number of defenders. The upper bound implies that if $h>0$, then $\min_{i\in\mathcal{D}} b_i > 0$, which ensures $b_i > 0$ for all EZs (the attacker remains outside every defender’s EZ). If all $b_i>0$, then $h$ is guaranteed to be positive if, $\min_{i\in\mathcal{D}} b_i > \beta \ln n,$ which provides a design condition on the smoothing parameter $\beta$, obtained from the lower bound of $h$ in \eqref{eqn:agg_safety_bounds}. Particularly, choosing $\beta$ sufficiently small makes $h$ a close approximation of $\min_i b_i$ (or closest engagement boundary to the attacker), thereby tightening the equivalence between the condition $h>0$ and the requirement $b_i>0~\forall~i \in \mathcal{D}$. On the other hand, selecting larger values of $\beta$ allows the aggregation to account for the combined effect of multiple EZs, rather than focusing solely on the nearest one.
\end{remark}
We can now define the safe set associated with the aggregated safety function as
\begin{align}
    \mathcal{C}_h \coloneqq \left\{(x_A, y_A, \gamma_A) \in \mathbb{R}^2 \times [0,2\pi) \;\big| h>0\right\}.
\end{align}
When the smoothing parameter $\beta \to 0$,  $\mathcal{C}_h\approx \mathcal{C}$ closely approximates the point-wise minimum of the individual margins. However, when $\beta$ is finite, the aggregation provides an inner approximation of the true safe set, that is, $\mathcal{C}_h\subset  \mathcal{C}$. 
\begin{figure}[h!]
    \centering    \includegraphics[width=0.85\linewidth]{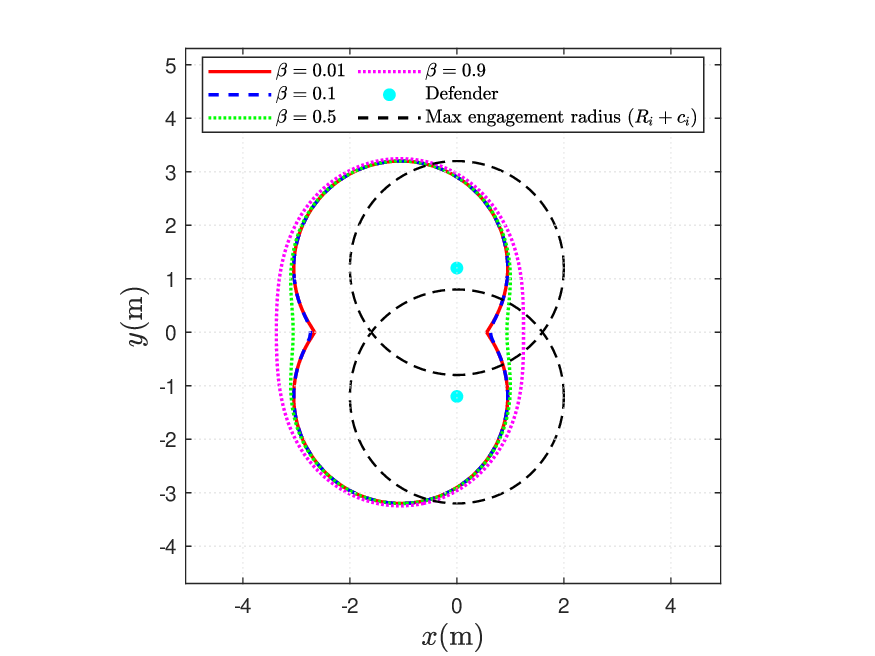}
    \caption{Variation of the safety function boundary $(h=0)$ for different choices of the parameter $\beta$.}
    \label{fig:var_bet}
\end{figure}

\Cref{fig:var_bet} compares the boundary defined by $h=0$ for two defenders under different values of the parameter $\beta$ for a fixed attacker's heading angle, where it can be observed that for small values of $\beta$, the aggregated safety function closely approximates the pointwise minimum of the individual safety functions. As a result, the overall unsafe region is essentially formed by the simple union of the individual EZs, producing a sharp concave notch between the two EZs. In practice, such a concave region may cause the attacker to enter unsafe regions, as it may not possess sufficient control authority to avoid it. In contrast, larger values of $\beta$ yield a smoother aggregation that results in a safety boundary that blends the influence of both EZs and eliminates the concave notch to produce a more regular and maneuverable safety region. Additionally, since $h$ aggregates the individual safety margins through a smooth log-sum-exp map, its time derivative is a weighted combination of the individual safety functions' rates, given as
\begin{align}
    \dot{h} &=\sum_{i\in \mathcal{D}}\dfrac{e^{-b_i/\beta}}{\sum_{i\in\mathcal{D}}e^{-b_i/\beta}}\dot{b}_i = \sum_{i\in \mathcal{D}}w_i\dot{b}_i \nonumber\\&=\underbrace{\sum_{i\in \mathcal{D}}w_i\left(v_i\cos\sigma_{iA} + \frac{v_i\nabla\rho_i}{r_{Ai}}\sin\sigma_{iA}-v_A\cos\sigma_{Ai}-\frac{v_A\nabla\rho_i}{r_{Ai}}\sin\sigma_{Ai} \right)}_{f_h}- \underbrace{{\sum_{i\in \mathcal{D}}w_i \nabla\rho_i}}_{g_h} \frac{a_A}{v_A}, \label{eqn:hdot}
\end{align}
obtained by differentiating \eqref{eq:h} with respect to time, where $\sum_{i\in\mathcal{D}}w_i=1$ with $w_i\geq0$. 

It can be observed from the above equation that $a_A$ also has relative degree one with respect to the aggregated safety function $h$. In this work, we explicitly account for the smooth lateral acceleration saturation model \eqref{eqn:sat_model}, making the commanded acceleration $a_A^C$ the control input to be designed. Consequently, an ideal design enforcing $h>0$ would require regulating $\ddot{h}$, which depends on the attacker’s dynamics as well as on the defender’s (typically unknown) maneuvering strategy. This could be verified by taking the derivative of $\dot{h}$ as in \eqref{eqn:hdot} with respect to time, where both $a_A^C$ and $a_i, \; i \in \mathcal{D}$ explicitly appear, confirming that $h$ has relative degree two with respect to $a_A^C$. 

To circumvent these issues and retain implementability under input constraints, we introduce a tightening parameter $\Delta>0$ and require the attacker to maintain $h\geq\Delta$ instead of $h>0$. This margin ensures that the attacker remains sufficiently far from every defender’s EZ, thereby preserving enough control authority to maneuver away even under bounded control inputs. To this end, the resulting tightened safe set is defined as
\begin{align}
    \mathcal{C}_\Delta \coloneqq \left\{(x_A, y_A, \gamma_A) \in \mathbb{R}^2 \times [0,2\pi) \;\big| h>\Delta,\; \Delta>0\right\}.
\end{align}
 which, by construction, satisfies $\mathcal{C}_{\Delta} \subset \mathcal{C}_h \subset \mathcal{C} \subset{\mathcal{E}}. $ To enable only the minimum necessary tightening of the safety constraint, we propose a time-varying parameter,
\begin{align}
    \Delta(t)
    \coloneqq \max\left\{0,\frac{1}{p_1}
    \left[
        |g_h(t)|\frac{a_{\max}}{v_A} - f_h(t)
    \right]\right\},
    \label{eq:Delta_final}
\end{align}
based on the worst-case evolution of safety boundaries under all admissible inputs. From \eqref{eqn:hdot}, the worst-case instantaneous rate 
of decrease of $h$ over all admissible accelerations is given by,
\begin{align}
    \dot{h}_{\min}(t)
    &= \min_{|a_A|\le a_{\max}} 
        \left\{ 
            f_h(t) - g_h(t)\frac{a_A}{v_A} 
        \right\}
        = f_h(t) - |g_h(t)|\,\frac{a_{\max}}{v_A},
    \label{eq:hdot_min_here}
\end{align}
since for $g_h(t)>0$ and  $g_h(t)<0$ the choices $a_A=a_{\max}$ and $a_A= -a_{\max}$, respectively, yield the largest negative contribution to $\dot{h}$ to produce the maximum possible decrease in the safety measure $h$. Additionally, from \eqref{eqn:sat_model}, a local approximation of the actuator response time under worst-case remaining control authority can be obtained as
\begin{align}
    \tau_\mathrm{resp} =\dfrac{1}{p_1},
\end{align}
since $\dot{a}_A\approx a_A^C - p_1a_A$ near saturation boundaries. This implies that $p_1$ represents the actuator’s effective response speed and $\tau_\mathrm{resp}$ represents the characteristic time required for the actual acceleration to track the commanded input $a_A^C$.
\begin{lemma}
[Safety Tightening Buffer]\label{lem:Delta}
Consider the dynamics of the safety constraint \eqref{eqn:hdot} and the designed tightening parameter as in \eqref{eq:Delta_final}. Then, for any $h(t_0)>\Delta(t_0), \forall t_0\geq 0$, the safety constraint remains non-negative over the response window $[t_0, t_0+\tau_\mathrm{resp}]$ for all  admissible control inputs, that is, 
\begin{align}
h(t)>0, \;    \forall t \in [t_0, t_0+\tau_\mathrm{resp}]. 
\end{align}    
\end{lemma}
\begin{proof}
From \eqref{eq:hdot_min_here}, it is evident that $\dot{h}(t)\geq \dot{h}_{\min}$, which implies that the instantaneous decrease in $h$ is lower bounded by worst case decrement $\dot{h}_{\min}$. On integrating the above inequality, we obtain
\begin{align}
    h(t)\geq h(t_0) + \int_{t_0}^t\dot{h}_{\min}(s)ds,  ~~t\in[t_0,t_0+\tau_{\mathrm{resp}}].
\end{align}
Since $\dot{h}_{\min}(s)\geq\dot{h}_{\min}(t_0)$ for all $s\in[t_0, t_0+\tau_\mathrm{resp}]$, we can simplify the above expression to obtain a conservative linear extrapolation safety function
\begin{align}
    h(t) \;\ge\; h(t_0) + (t-t_0)\,\dot h_{\min}(t_0), 
    ~~ t\in[t_0,t_0+\tau_{\mathrm{resp}}], \label{eqn:h_approx}
\end{align}
over a short interval of length $\tau_{\mathrm{resp}}$. 

If $\dot h_{\min}(t_0) \ge 0$, then $h(t)\ge h(t_0)\ge 0$ on 
$[t_0,t_0+\tau_{\mathrm{resp}}]$ which implies that $h(t)$ remains non-negative over the finite response window. However, when $\dot{h}_{\min}(t_0)\leq0$, we have
\begin{align}
    h(t_0 +\tau_\mathrm{resp}) \geq h(t_0) + \tau_\mathrm{resp}\dot{h}_{\min}(t_0)
\end{align}
using \eqref{eqn:h_approx}. The above expression can be rewritten using $h(t_0)>\Delta(t_0)$ as
\begin{align}
    h(t) \geq \Delta(t_0) + \tau_\mathrm{resp} \dot{h}_{\min}(t_0),
    \end{align}
    which provides the condition on the $\Delta(t)$ as proposed in \eqref{eq:Delta_final} to ensure  $h(t)>0$. This concludes the proof.
\end{proof}
\begin{figure}
    \centering    \includegraphics[width=0.8\linewidth]{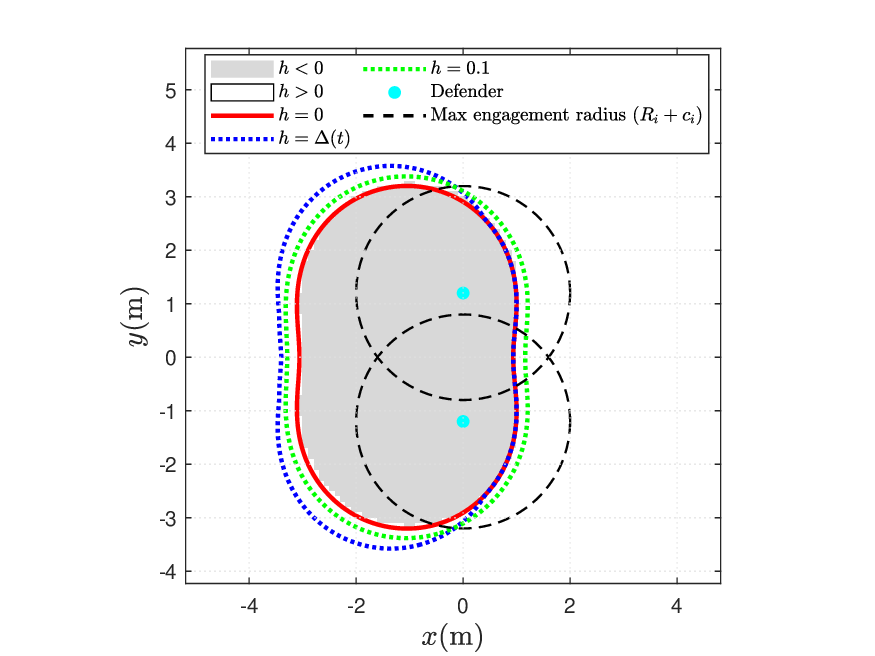}
    \caption{Variation of the safety function boundary with tightening parameter $\Delta$.}
    \label{fig:var_del}
\end{figure}

\Cref{fig:var_del} compares the engagement boundary under different tightening strategies of the aggregate safety function with two defenders at a constant attacker's heading ($\gamma_A=0 ^\circ$). The red solid line represents the nominal engagement boundary corresponding to zero tightening corresponding to $\Delta(t)=0$, while the green dotted line represents the engagement boundary under a constant tightening margin corresponding to $\Delta(t)=0.1$ that uniformly enlarges the unsafe region in all directions relative to the nominal case. In comparison, the proposed state-dependent tightening parameter expands the nominal unsafe region non-uniformly as depicted by the blue dotted boundary. Specifically, the unsafe set is enlarged on the left side of the nominal engagement boundary, while no expansion occurs on the right side. This asymmetry is consistent with the attacker’s heading $\gamma_A = 0^\circ$ since when approaching from the left, the nominal engagement boundary lies on a collision trajectory that necessitates tightening. However, the same heading cannot drive the attacker into the EZ for positions to the right of the nominal boundary and therefore does not require any expansion of the engagement boundary. 
\begin{remark}
    Since the proposed tightening function $\Delta(t)$ in  \eqref{eq:Delta_final} utilizes smooth functions $f_h$ and $g_h$ that are dependent on the relative kinematics and vehicle states, the time derivative of the tightening function is bounded, that is, $|\dot{\Delta}(t)|<L_\Delta$, for some finite constant $L_\Delta$>0.
    \label{rem:der_tighten}
\end{remark}

Up to this point, the safety constraints have been characterized, and a tightening mechanism has been introduced to enlarge the effective unsafe set, ensuring that worst-case system behavior cannot drive the attacker into the actual EZs. In the considered problem, the attacker must simultaneously pursue target interception while avoiding safety violations, which naturally introduces competing objectives. To address this, we develop a switching condition and a corresponding guidance law that allow the attacker to transition between target-interception and EZ-avoidance modes. The proposed lateral acceleration component that steers the attacker away from all defender-induced engagement zones is given by
\begin{align}
    a_A^b = \dfrac{v_AK_s\left(h-\Delta\right)-v_A\displaystyle\sum_{i\in \mathcal{D}}w_i\left(v_i\cos\sigma_{iA} + \dfrac{v_i\nabla\rho_i}{r_{Ai}}\sin\sigma_{iA}-v_A\cos\sigma_{Ai}-\dfrac{v_A\nabla\rho_i}{r_{Ai}}\sin\sigma_{Ai}\right)}{\displaystyle\sum_{i\in\mathcal{D}}w_i\nabla\rho_i}, \label{eqn:safe_ctrl}
\end{align}
where $K_s$ denotes the controller gain to be designed. This component is obtained by imposing $\dot{h}=-K_s\left(h-\Delta(t)\right)$ that guarantees forward invariance of the tightened safe set $\mathcal{C}_{\Delta}$, thus ensuring that the attacker remains outside all defender-induced EZs while the safety margin decays in a controlled manner. 

For target interception, we propose the attacker's lateral acceleration component as
\begin{align}
a_A^{T} = -K_I v_A \sigma_{AT} - \dfrac{v_A^2 \sin \sigma_{AT}}
{r_{AT}} \label{eqn:inter_law},
\end{align}
where $K_I$ is the controller gain for interception to be designed. The above component of the lateral acceleration is obtained by ensuring $\dot{\sigma}_{AT}=-K_1\sigma_{AT}$, which aligns the attacker's velocity along the attacker-to-target LOS in order to steer the attacker towards the target. This results in pure pursuit behavior, which is sufficient to guarantee interception of a stationary target. 

To determine when the attacker may violate the safety constraint, we define the boundary based on the aggregate safety function rate as
\begin{align}
   \psi(a_A) =f_h + g_h\dfrac{a_A}{v_A}+\alpha_h(h-\Delta)=0, \label{eqn:switch_h_dot_cond}
\end{align}
where $\alpha_h(\cdot)$ is an extended class-$\mathcal{K}$ function\footnote{An extended class-$\mathcal{K}$ function is a continuous and strictly increasing function 
$\phi:(-\ell_1,\ell_2)\to\mathbb{R}$ satisfying $\phi(0)=0$, where $\ell_1,\ell_2>0$. 
Unlike a standard class-$\mathcal{K}$ function, whose domain is $[0,\infty)$, 
an extended class-$\mathcal{K}$ function is defined on an interval containing the origin 
and therefore permits negative arguments.}. The first two terms in \eqref{eqn:switch_h_dot_cond} correspond to $\dot{h}(a_A)$, which represent how $a_A$ influences the instantaneous rate of change of the safety function. The last term is always a positive quantity that acts as a corrective barrier, preventing the safety function from decreasing too rapidly. Therefore, the $\psi(a_A)=0$ represents the boundary at which the evolution of the safety function exactly balances the barrier correction term. Consequently, $\psi(a_A)> 0$ indicates that the attacker's applied lateral acceleration satisfies the safety constraint, whereas $\psi(a_A)\leq0$ indicates that the applied control input would lead to violation of the safe set. In this work, we select $\alpha_h = K_s\left(h-\Delta\right)$, consistent
with the safety feedback term in \eqref{eqn:safe_ctrl}, so that purely applying this term ensures that the tightened safety condition is never violated.
To handle the trade-off between safety and target-interception, we blend $a_A^b$ and $a_A^T$ via a switching function/convex combination of the form
\begin{align}
    a_A^d = \alpha\left(\psi\left(a_A^T\right)\right) a_A^b + \left(1-\alpha\left(\psi\left(a_A^T\right)\right)\right)a_A^T,
    \label{eqn:blend_ctrl}
\end{align}
where $\alpha: \mathbb{R}\to[0,1]$ is a continuous, locally Lipschitz, and monotonically decreasing switching function. Thus there exists $\varepsilon>0$ such that
\begin{align}
    \alpha\left(\psi\left(a_A^T\right)\right) = \begin{cases}
        1, &\text{if}\;\psi\left(a_A^T\right)\leq0,\\
        0, &\text{if}\;\psi\left(a_A^T\right)\geq\varepsilon
    \end{cases}. \label{eqn:switch}
\end{align}
For brevity, we denote $\psi\coloneqq\psi\left(a_A^T\right)$ and $\alpha\coloneqq\alpha\left(\psi\left(a_A^T\right)\right)$, representing $\psi$ evaluated at $a_A^T$ and $\alpha$ evaluated at $\psi$, respectively, unless otherwise stated. 

Therefore, based on the switching strategy in \eqref{eqn:switch}, the attacker's state-space can be divided into three regions, as presented in \Cref{tb:regions}. 
\begin{table}[h!]
\centering
\begin{tabular}{p{1cm} p{2.2cm} p{2.4cm} p{2.6cm} p{6.2cm}}
\toprule
\multicolumn{1}{c}{\textbf{Region}} &
\multicolumn{1}{c}{\textbf{Mode}} &
\multicolumn{1}{c}{\textbf{Condition on $\alpha$}} &
\multicolumn{1}{c}{\textbf{Guidance Law}} &
\multicolumn{1}{c}{\textbf{Interpretation}} \\
\midrule
I
& Pure target interception 
& $\alpha = 0$ 
& $a_A^d = a_A^T$ 
& The attacker is sufficiently far from all engagement zones; the guidance law focuses entirely on target interception. \\

II
& Pure EZ avoidance 
& $\alpha = 1$ 
& $a_A^d = a_A^b$ 
& The attacker is near an engagement boundary; the safety component is fully activated to prevent entry into the engagement zone. \\

III
& Blended mode 
& $0<\alpha<1$ 
& $a_A^d = (1-\alpha)a_A^T + \alpha a_A^b$ 
& Both interception and safety terms contribute to the commanded acceleration. This region represents the transition between interception and safety dominance, where $\alpha$ varies continuously due to its locally Lipschitz and nonincreasing properties. \\

\bottomrule
\end{tabular}
\caption{Operating regions of the proposed guidance law.}
\label{tb:regions}
\end{table}

\begin{assumption} 
The time derivative of $\alpha$ is uniformly bounded, that is, there exists $L_\alpha>0$ such that $|\dot{\alpha}|<L_\alpha~\forall~t\geq 0$, ensuring smooth transition between various guidance modes. \label{asm:bound_switch}
\end{assumption}
\begin{lemma}\label{lem:safety}
   The switching condition in \eqref{eqn:switch} defines the boundary $\psi(a_A^T)=0$, 
which identifies the region beyond which applying only the target-interception 
lateral acceleration in \eqref{eqn:inter_law} may compromise safety. 
This condition, therefore, characterizes the threshold beyond which the desired lateral acceleration $a_A^d$ can potentially violate the tightened safety 
requirement $h(t)>\Delta(t)$.
\end{lemma}
\begin{proof}
  Using \eqref{eqn:hdot}, we obtain the dynamics of the aggregated safety function at $a_A=a_T$ as
  \begin{align}
      \dot{h}({a_T}) = f_h + g_h\frac{a_T}{v_A}.
  \end{align}
  To guarantee that $\dot{h}({a_T})$ above does not violate the tightened safety condition, one must ensure the condition,
  \begin{align}
      \dot{h}(a_T) +\alpha_{h}(h-\Delta)=\psi(a_A^T) = f_h + g_h\frac{a_T}{v_A}  +\alpha_{h}(h-\Delta) \geq0,  \end{align}
      which on comparison with \eqref{eqn:switch_h_dot_cond} provides the switching boundary $\psi(a_T)=0$ as presented in \eqref{eqn:switch}. Moreover, it is important to note that $a_A^b$ satisfies $\psi(a_A^b)=0$ by construction, obtained by ensuring $\dot{h}=-K_s\left(h-\Delta(t)\right)$. Since $\psi$ is affine in $a_A$, the proposed blended control in \eqref{eqn:blend_ctrl} will result in
        \begin{align}
            \psi(a_A^d)=\left(1-\alpha\right)\psi(a_T),\;\forall\; \alpha \in [0,1]. \label{eqn:res_bound}
        \end{align}
      Using the above expression, we can analyze the safety condition in the three regions of the state-space. In Region I, when $\psi(a_A^T)>\epsilon$, we have $\alpha=0$ and therefore $\psi(a_A^d)=\psi(a_A^T)$. Since $\psi(a_A^T)\geq0$, it follows that $\psi(a_A^d)\geq0$, ensuring safety. In Region II, the safety component \eqref{eqn:safe_ctrl} is fully active. When $\psi(a_A^T)\leq0$ and $\alpha=1$, we obtain $\psi(a_A^d)=0$. In Region III, when $0<\psi(a_A^T)<\epsilon$, we have $0<\alpha<1$. Using \eqref{eqn:res_bound}, it follows that $\psi(a_A^d)>0$ since $\psi(a_A^T)>0$.

Therefore, ensuring $\psi(a_A^T)\geq0$ is sufficient to guarantee safety $\psi(a_A^d)\geq0$ across all three regions. Under this condition, $\psi(a_A^T)\geq0$ implies that the aggregate safety function has a nonnegative rate of change, so the safety margin $h-\Delta$ is nondecreasing and the constraint $h>\Delta$ is preserved. On the other hand, when $\psi(a_A^T)<0$, the aggregate safety function decreases, which may drive $h$ to $\Delta$ in finite time and violate the safety condition $h(t)>\Delta$. This concludes the proof.      
 \end{proof}

To account for the lateral acceleration dynamics and the constraints in the permissible limits, we define the tracking error between the actual and the desired lateral acceleration as, 
\begin{align}
    z =a_A-a_A^d, \label{eqn:pseudi_ctrl}
\end{align}
representing the offset between the actual and the desired values while representing more realistic dynamics. Driving this tracking error to zero, that is, $\lim_{t \to \infty}z\to 0$, ensures that the attacker accurately follows the desired guidance strategy in \eqref{eqn:blend_ctrl}. To ensure this, we design the commanded lateral acceleration for the attacker as
\begin{align}
    a_A^c =& \dfrac{p_1a_A + \dot{a}_A^d+(\alpha-1)\dfrac{\sigma_{AT}}{v_A} +\alpha \left(h-\Delta\right)\dfrac{f_h}{v_A} - K_az}{1-\left(\dfrac{a_A}{a_{\max}}\right)^n} ,\label{eqn:comd_acc}
\end{align}
where $K_a\in\mathbb{R}_{> 0}$ is the controller gain. The first two terms in the numerator of $a_A^c$ compensate for the lateral acceleration dynamics and the time variation in the desired lateral acceleration. The third and fourth terms capture the contributions from the interception and safety guidance objectives blended via the switching function and represent the backstepping terms from the recursion. The last term in the numerator $a_A^c$ represents the corrective input that drives the acceleration tracking error to zero.
\begin{theorem}
    Consider the target–attacker–defender relative kinematics \eqref{eqn:rel_dyn_1}–\eqref{eqn:rel_dyn_2} with the input saturation model \eqref{eqn:sat_model}. Suppose the attacker’s commanded lateral acceleration is given by \eqref{eqn:comd_acc}, with the resulting actual lateral acceleration governed by \eqref{eqn:blend_ctrl}. Then the attacker intercepts the target while remaining outside the defenders' engagement boundaries and respecting the actuator bounds, provided the controller gains satisfy
    \begin{align}    K_I>\sup_{\alpha\in(0,1)}\frac{2v_A+\epsilon_\sigma + L_\alpha}{
            2\left(R_i+C_i\right)(1-\alpha)^2} + w_1,\; K_s>\sup_{\alpha\in(0,1)}\frac{\epsilon_h+\epsilon_\Delta+L_\alpha}{2\alpha^2} +w_1, \label{eqn:suff_cond}
    \end{align}
    where $\epsilon_\sigma>0, \epsilon_h>0, \epsilon_\Delta>0$ and $ w_1>0$ are constants. \label{thm:thm_1}
\end{theorem}
\begin{proof}
    Consider the candidate for the Lyapunov function
    \begin{align}
        V_1=\frac{1}{2}\left(1-\alpha\right)\sigma_{AT}^2 + \frac{1}{2}\alpha \left(h-\Delta\left(t\right)\right)^2 .
    \end{align}
    Differentiating $V_1$ with respect to time and using \eqref{eqn:bdot} and \eqref{eqn:hdot}, we obtain,
    \begin{align}
        \dot{V}_1 =& (1-\alpha)\sigma_{AT}\dot{\sigma}_{AT} + \alpha \left(h-\Delta\right)\left(\dot{h}-\dot{\Delta}(t)\right) + \frac{\dot{\alpha}}{2}\left(\left(h-\Delta(t)\right)^2-\sigma_{AT}^2\right) \nonumber 
         \\
        =&\left(1-\alpha\right)\sigma_{AT}\left(\frac{a_A^d}{v_A}+\frac{v_A}{r_{AT}}\sin\sigma_{AT}\right) +\alpha\left(h-\Delta\right)\left(f_h + g_h \frac{a_A^d}{v_A}\right) +\left((1-\alpha)\frac{\sigma_{AT}}{v_A} +\alpha \left(h-\Delta(t)\right)\frac{f_h}{v_A}\right) z \nonumber\\&- \alpha \left(h-\Delta\right)\dot{\Delta}(t) + \frac{\dot{\alpha}}{2}\left(\left(h-\Delta(t)\right)^2-\sigma_{AT}^2\right)
    \end{align}
    using $a_A=z + a_A^d$ from \eqref{eqn:pseudi_ctrl}. On substituting the desired lateral acceleration \eqref{eqn:blend_ctrl} in the above equation, we obtain 
    \begin{align}
        \dot{V}_1 =& -(1-\alpha)^2K_I\sigma_{AT}^2 -\alpha^2K_s\left(h-\Delta(t)\right)^2 +\left((1-\alpha)\frac{\sigma_{AT}}{v_A} +\alpha \left(h-\Delta(t)\right)\frac{f_h}{v_A}\right) z \nonumber\\&+\alpha(1-\alpha)\Bigg[\sigma_{AT}\left(a_A^b+\frac{v_A}{r_{AT}}\sin\sigma_{AT}\right)+\left(h-\Delta\right)\left(f_h+\frac{g_h}{v_A}a_A^T\right) \Bigg]- \alpha \left(h-\Delta\right)\dot{\Delta}(t) \nonumber \\&+ \frac{\dot{\alpha}}{2}\left(\left(h-\Delta(t)\right)^2-\sigma_{AT}^2\right). \label{eqn:dot_v_1}
    \end{align}
     To explicitly account for the actuator dynamics introduced by the saturation model in \eqref{eqn:sat_model}, we augment the Lyapunov function with an additional term that captures the acceleration tracking error, $V=V_1+V_2$, where $V_2 =z^2/2$. On differentiating $V$ with respect to time and using \eqref{eqn:sat_model}, \eqref{eqn:dot_v_1}, we obtain
    \begin{align}
        \dot{V}=& \dot{V}_1 + \dot{V}_2 = \dot{V_1}+z\dot{z}=\dot{V_1}+z\left(\dot{a}_A-\dot{a}_A^d\right) \nonumber\\ =&-(1-\alpha)^2K_1\sigma_{AT}^2 -\alpha^2K_s\left(h-\Delta(t)\right)^2 +\left((1-\alpha)\frac{\sigma_{AT}}{v_A} +\alpha \left(h-\Delta(t)\right)\frac{f_h}{v_A}\right) z \nonumber\\&+\alpha(1-\alpha)\Bigg[\sigma_{AT}\left(a_A^b+\frac{v_A}{r_{AT}}\sin\sigma_{AT}\right)+\left(h-\Delta\right)\left(f_h+\frac{g_h}{v_A}a_A^T\right) \Bigg]- \alpha \left(h-\Delta\right)\dot{\Delta}(t) \nonumber \\&+ \frac{\dot{\alpha}}{2}\left(\left(h-\Delta(t)\right)^2-\sigma_{AT}^2\right)+
      \Bigg[\left[1-\left(\frac{a_A}{a_{\max}}\right)^n\right]a_A^c - p_1 a_A -\dot{a}_A^d\Bigg] z .
    \end{align}
    Choosing the commanded acceleration proposed in \eqref{eqn:comd_acc} renders the derivative of the Lyapunov function candidate as
    \begin{align}
        \dot{V}=&-(1-\alpha)^2K_I\sigma_{AT}^2 -\alpha^2K_s\left(h-\Delta(t)\right)^2 -K_az^2 + \frac{\dot{\alpha}}{2}\left(\left(h-\Delta(t)\right)^2-\sigma_{AT}^2\right)- \alpha \left(h-\Delta\right)\dot{\Delta}(t)\nonumber\\&+\alpha(1-\alpha)\Bigg[\sigma_{AT}\left(a_A^b+\frac{v_A}{r_{AT}}\sin\sigma_{AT}\right)+\left(h-\Delta\right)\left(f_h+\frac{g_h}{v_A}a_A^T\right) \Bigg]. \label{eqn:im_dotv}
    \end{align}
    We analyze the safety guarantees and stability of the proposed guidance law \eqref{eqn:comd_acc} by partitioning the state-space into three distinct regions according to the switching function in \eqref{eqn:switch}. 
    In \textbf{Region I (pure target interception mode)}, $\alpha=0$ and $\psi(a_A^T)\geq\epsilon$, which implies the target interception term $a_A^T$ is purely activated in $a_A^d$. Under such a condition $\dot{\alpha}=0$, $\alpha\left(1-\alpha\right)=0$ and $\dot{\Delta}=0$ renders the Lyapunov derivative as
    \begin{align}
        \dot{V} = -K_I\sigma_{AT}^2 - K_a z^2 <0, \;\forall \;(\sigma_{AT}, z) \in \mathbb{R}^2 \setminus(0,0),
    \end{align}
    if $K_1>0$ and $K_a>0$. This implies that both $\sigma_{AT}$ and $z$ converge asymptotically to zero. The term $\sigma_{AT} \to 0$ implies the attacker aligns its heading along the LOS to the target, resulting in $\dot{r}_{AT}=-v_A<0$, and guaranteeing monotonic range reduction similar to pure-pursuit guidance. In addition, $z \to 0$ results in the attacker perfectly tracking the desired acceleration, that is, $a_A=a_A^d=a_A^T$. 
    
    In \textbf{Region II (pure EZ avoidance mode)}, $\psi(a_A^T)\leq0$ and $\alpha=1$, and only the safety component $a_A^b$ is fully active in the desired acceleration $a_A^d$. Since under such conditions $\dot{\alpha}=0$ and $\alpha\left(1-\alpha\right)=0$, the Lyapunov derivative in \eqref{eqn:im_dotv} simplifies to
    \begin{align}
        \dot{V} = -K_s\left(h-\Delta\left(t\right)\right)^2 - K_a z^2-\left(h-\Delta\right)\dot{\Delta}(t)\leq-\frac{K_s}{2}\left(h-\Delta\right)^2 - K_az^2 + \frac{L_\Delta^2}{2K_s}
    \end{align}
    following \Cref{rem:der_tighten}. It follows from the above expression that if $K_s>0$ and $K_a>0$, then decrement of $V$ is guaranteed outside the compact set $$\Omega_{II}\coloneqq\left\{(h-\Delta, z)\in \mathbb{R}^2\middle|\dfrac{K_s}{2}\left(h-\Delta\right)^2+K_az^2\leq\dfrac{L_\Delta^2}{2K_s}\right\}.$$ Therefore, $h-\Delta$ and $z$ are uniformly ultimately bounded in {Region II} with the ultimate performance bounds given as $|h-\Delta|\leq\dfrac{L_\Delta}{K_s}$ and $|z|\leq\sqrt{\dfrac{L_\Delta^2}{2K_2K_a}}$. This results in $h\to\Delta+\epsilon_1$ and $z\to \epsilon_2$, where the residual errors $\epsilon_1$ and $\epsilon_2$ are bounded by the ultimate performance limits as derived above and can be made arbitrarily small by a suitable choice of the design parameters. This implies that in {Region II}, the guidance law prioritizes safety and steers the attacker toward a small neighborhood of the tightened safe set. Similarly, the attacker's lateral acceleration remains within a small neighborhood of the desired acceleration, that is, $a_A\to a_A^d +\epsilon_2$.

    In \textbf{Region III (blended mode)}, $0<\psi(a_A^T)<\epsilon$, $0<\alpha<1$, such that a convex combination of $a_A^b$ and $a_A^T$ is active in $a_A^d$, which on following  \Cref{rem:der_tighten} and \Cref{asm:bound_switch}, simplifies the Lyapunov derivative in \eqref{eqn:im_dotv} to
     \begin{align}
       \dot{V}_2 \leq & -\left[(1-\alpha)^2K_I-\left(\frac{v_A}{r_{AT}}+ \frac{\epsilon_\sigma}{2}\right)-\frac{L_\alpha}{2}\right]\sigma_{AT}^2 -\left[\alpha^2K_s - \frac{\epsilon_h +\epsilon_\Delta + L_\alpha}{2}\right]\left(h-\Delta(t)\right)^2 -K_az^2 \nonumber\\&+ \frac{1}{2\epsilon_\sigma}|a_A^b|^2 + \frac{1}{2\epsilon_h}|\phi|^2 + \frac{1}{2\epsilon_\Delta}L_\Delta,
       \end{align}
      where $\phi=f_h+\dfrac{g_h}{v_A}a_A^T$. The above expression provides the sufficient condition as presented in \eqref{eqn:suff_cond} on the gains (after substituting $r_{AT}\geq R_i+c_i$, since in Region III, interception cannot occur and therefore $r_{AT}$ is bounded away from zero) to ensure that the first two terms in the above expression are always negative definite, yielding
    \begin{align}
        \dot{V} = & -w_1\sigma_{AT}^2-w_1\left(h-\Delta\right)^2 -K_az^2 + \overline{d},
    \end{align}
    where $\overline{d}= \dfrac{1}{2\epsilon_\sigma}|a_A^b|^2 + \dfrac{1}{2\epsilon_h}|\phi|^2 + \dfrac{1}{2\epsilon_\Delta}|\dot{\Delta}|$ denotes the upper bound on the residual terms. It follows from the above expression that the Lyapunov function candidate $V$ decreases outside the compact set $$\Omega_{III}\coloneqq\Big\{(\sigma_{AT}, h-\Delta,z)\in \mathbb{R}^3\big|w_1\sigma_{AT}^2 +w_2\left(h-\Delta\right)^2 + K_az^2\leq\overline{d}\Big\},$$ with ultimate performance bounds, $|\sigma_{AT}|, |h-\Delta|\leq \sqrt{\dfrac{\overline{d}}{w_1}}$ and $|z|\leq \sqrt{\dfrac{\overline{d}}{K_a}}$. This implies that in {Region III}, $\sigma_{AT}\to \epsilon_3$, $h\to \Delta +\epsilon_4$ and $z\to\epsilon_5$, where $\epsilon_3, \epsilon_4$ and $\epsilon_5$ are positive constants. This results in the attacker maintaining a bounded LOS misalignment $\sigma_{AT} \not\to 0$ within a bounded neighborhood of the tightened safety boundary, while the attacker's lateral acceleration tracks the desired acceleration up to a bounded residual error $a_A\to a_A^d +\epsilon_5$. 

    Combining the analyses across the three regions establishes the desired safety and convergence properties. In Region I, the closed-loop dynamics are asymptotically stable to guarantee target interception. In Regions II and III, the Lyapunov derivative is negative outside compact sets, implying uniform ultimate boundedness of $\sigma_{AT}, h-\Delta(t)$ and $z$. As a consequence, the attacker remains within a bounded neighborhood of the tightened safety boundary while never penetrating the defender-induced EZs, and ensuring that its lateral acceleration tracks the desired values up to a bounded residual error. Moreover, the smooth saturation dynamics in \eqref{eqn:sat_model} ensure that the actual control input $a_A$ remains within the admissible limits for all time. Therefore, under the gain conditions in \eqref{eqn:suff_cond}, the proposed guidance law guarantees target interception, practical safety from the defenders, while respecting the constraints on the attacker's control inputs. This completes the proof.      
\end{proof}

\begin{remark}
    The ultimate bounds derived in Regions II and III can be made smaller by choosing high values of the gains $K_1$, $K_S$, and $K_a$, improving the uniform ultimate boundedness  performance and shrinking the invariant compact sets. Particularly, increasing $K_s$ reduces the steady-state deviation of $h-\Delta$ from zero, and larger $K_a$ decreases the residual acceleration tracking error $z$. However, excessively large gains may amplify high-amplitude oscillations, leading to degraded performance. Therefore, the gains must be selected to balance tight ultimate bounds with acceptable transient smoothness and actuator limitations.
\end{remark}
\begin{remark}
    From \Cref{thm:thm_1}, it is evident that both the actual lateral acceleration $a_A$ and the acceleration tracking error $z$ are uniformly ultimately bounded. Since $a_A^d=a_A+z$, it follows that the desired lateral acceleration $a_A^d$ is also bounded. As a consequence, the desired lateral acceleration remains physically realizable, and the derivative $\dot{a}_A^d$ associated with the proposed control in \eqref{eqn:comd_acc} remains bounded. This ensures that all internal signals in the backstepping-based command law are well-defined and prevents amplification of high-frequency dynamics in the closed-loop system.
\end{remark}
\begin{remark}
    It is important to emphasize that both the desired acceleration \eqref{eqn:blend_ctrl} and the commanded acceleration \eqref{eqn:comd_acc} rely only on relative information and do not require the information of control input of other agents. Compared to prior approaches, the proposed method yields a fully distributed, scalable method for practical multi-agent engagement scenarios, robust to communication-induced anomalies or adversarial information constraints, even when inter-agent communication is limited or unreliable. Furthermore, the independence from adversarial control inputs ensures robustness to changes in defender strategies (as well as to the number of defenders), eliminating the need to retune or restructure the attacker’s guidance law.
\end{remark}

\section{Maximum Engagement-Range-based Input-Constrained Intercept Guidance Design}
In this section, we design the guidance law for the attacker that enforces safety through a fixed stand-off distance equal to the defender's maximum engagement range constraint $R_i+c_i$  while ensuring prescribed control input bounds. Such a design provides a baseline for comparison with the proposed EZ-based safety constraints. In contrast to prior works \cite{doi:10.2514/1.37030,doi:10.2514/1.G003157,10839025,doi:10.2514/1.G003223,doi:10.2514/6.2026-0121} that employ similar range-based safety constraints but neglect actuator limits, the baseline design explicitly incorporates input bounds to enable a fair comparison. 

For notational convenience, we adopt the same variables used in the EZ-based formulation and denote them with the superscript $(\cdot)^{MR}$ to represent the quantities associated with the maximum engagement-range-based formulation. The goal here is to design the attacker's guidance strategy to satisfy \eqref{eqn:cns_saf_cond} such that it maintains a position outside the defender's maximum range. Under such considerations, the conservative safety and the aggregated conservative safety function are given by
\begin{align}
b_i^{MR}=r_{Ai}-R_i-c_i, \;\;h^{MR}= -\beta\log\left(\sum_{i \in\mathcal{D}}e^{-{b_t}^{MR}/\beta}\right), \;\forall\; i \in \mathcal{D},   
\end{align} 
where $\beta>0$ is a constant. 

From the above safety measure, $b_i^{MR}\geq0$  represents the scenarios when the attacker is outside the $i$\textsuperscript{th} defender's maximum engagement range, that is, $r_{AD_i}>R_i+c_i$. In comparison, $b_i^{MR}\leq 0$ indicates scenarios when the attacker is inside the $i$\textsuperscript{th} defender's maximum engagement range. Additionally, similar to \Cref{rem:agg_safety_cond}, $h^{MR}$ aggregates the safety measure from all the defenders using a log-sum-exp operator (soft-min) function, with the parameter $\beta$ governing how strongly the aggregated function emphasizes the most critical defender. Therefore, we can define the safe set that encompasses the region outside the maximum engagement range of all the defenders as
\begin{align}
    \mathcal{C}_h^{MR} \coloneqq \left\{(x_A, y_A, \gamma_A) \in \mathbb{R}^2 \times [0,2\pi) \;\big|~ h^{MR}>0\right\}.
\end{align}
Thus, remaining outside the above set, that is, $h^{MR}>0$, guarantees that $b_i^{MR}>0,\forall\ i\in\mathcal{D}$, using inferences from \Cref{rem:agg_safety_cond}.

On differentiating $h^{MR}$ with respect to time and using \eqref{eqn:rel_dyn_1}, we obtain the dynamics of the aggregate conservative safety measure as
\begin{align}
     \dot{h}^{MR} &= \sum_{i\in \mathcal{D}}w_i^{MR}\dot{r}_{Ai} =  \sum_{i\in \mathcal{D}}w_i^{MR}\left(v_i\cos\sigma_{iA}-v_A\cos\sigma_{Ai}\right)\label{eq:hdotmr}
\end{align}
where $w_i^{MR}=\dfrac{e^{-b_i^{MR}/\beta}}{\sum_{i\in\mathcal{D}}e^{-b_i^{MR}/\beta}}$ denotes the normalized soft-min weight associated with each defender. One can observe from \eqref{eq:hdotmr} that the first derivative of the conservative safety function depends only on the relative velocity between the attacker and the defenders. Since the attacker’s lateral acceleration does not appear explicitly in $\dot{h}^{MR}$, we further differentiate $\dot{h}^{MR}$ with respect to time to obtain
\begin{align}
     \ddot{h}^{MR}&=\underbrace{\frac{1}{\beta}
   \left(
      (\dot{h}^{MR})^{2}
      - \sum_{i\in\mathcal{D}} w_i^{MR} (\dot{b}_i^{MR})^{2}
   \right)
   + \sum_{i\in\mathcal{D}} w_i^{MR} \left(-a_i\sin\sigma_{iA} + r_{Ai}\dot{\theta}_{Ai}^2\right)}_{f_h^{MR}} +\underbrace{\left(\sum_{i\in\mathcal{D}} w_i\sin\sigma_{Ai} \right)}_{g_h^{MR}}a_A, \label{eqn:ddot_h_MR}
\end{align}
using \eqref{eqn:rel_dyn_1} and \eqref{eqn:rel_dyn_2}. The drift and input coupling terms in \eqref{eqn:ddot_h_MR} are shown in underbraces. It is important to observe the above expression that the attacker's lateral acceleration has a relative degree of two with respect to the conservative safety constraint $h^{MR}$.

Ensuring forward invariance of the set $\mathcal{C}_h^{MR}$ ensures safety by forcing the attacker to remain outside the maximum engagement range of all defenders. Since $h^{MR}$ has a relative degree of two, we introduce an auxiliary function to enable the higher-order control barrier function-based design, given as
\begin{align}
\psi_1^{MR}(x) = \dot{h}^{MR}(x) + \alpha_1\!\left(h^{MR}(x)\right) = \dot{h}^{MR}(x) + K_1^{MR}h^{MR} , \label{eqn:hocbf_cnd}
\end{align}
where $\alpha_1(\cdot)$ is an extended class-$\mathcal{K}$ function. In this work, we utilize a linear class-$\mathcal{K}$ function  of the form $\alpha_1\left(h^{MR}\right)=K_1^{MR}h^{MR}$, where $K_1^{MR}>0$ . The auxiliary function \eqref{eqn:hocbf_cnd} incorporates both the conservative safety function and its derivative to enable the enforcement of the safety constraints for relative-degree-two systems. The admissible set corresponding to the auxiliary constraint is defined as
\begin{align}
\mathcal{C}_1^{MR} = \left\{ x \in \mathbb{R}^n \mid \psi_1^{MR}(x) \ge 0 \right\}.
\end{align}
It follows that enforcing the forward invariance of the sets $\mathcal{C}_1^{MR}$ will ensure the forward invariance of the set $\mathcal{C}_h^{MR}$, thus guaranteeing that the attacker maintains a conservative safety constraint throughout its maneuver. 

To retain implementability under input constraints, we introduce a tightening parameter $\Delta^{MR}>0$ and ensure safety by satisfying $\psi_1^{MR}>\Delta^{MR}$ instead of $\psi_1^{MR}>0$. This tightening contracts the nominal safe set and introduces a safety buffer that preserves forward invariance of the conservative safety constraint despite bounded control authority. This tightening parameter is given as
\begin{align}
\Delta^{MR}(t)
\coloneqq
\max\left\{0,\frac{1}{p_1}
\left[
|g_h^{MR}|\,a_{\max}
-
f_h^{MR}
-
K_1^{MR}\dot{h}^{MR}
\right]\right\},
\label{eq:Delta_MR}
\end{align}
where $f_h,g_h$ are defined in \eqref{eqn:ddot_h_MR}. As a consequence, we can define the tightened conservative safe set as
\begin{align}
    \mathcal{C}_\Delta^{MR} = \left\{ x \in \mathbb{R}^n \mid \psi_1^{MR}(x) \ge \Delta^{MR} \right\}.
\end{align}
Hence, ensuring forward invariance to the set $\mathcal{C}_\Delta^{MR}$ will ensure attacker safety at all times, since, by construction, $ \mathcal{C}_\Delta^{MR}\subset \mathcal{C}_1^{MR}\subset \mathcal{C}_h^{MR}\subset \mathcal{E}^{MR}$. 
\begin{remark}
    Similar to \Cref{lem:Delta}, we can show that the tightening parameter \eqref{eq:Delta_MR}  is designed to counteract the worst-case decrease of the auxiliary barrier function $\psi_1^{MR}$ over the actuator response interval. In particular, the term $|g_h^{MR}|a_{\max}$ captures the largest possible adverse contribution of the bounded lateral acceleration to the safety dynamics, whereas the terms $f_h^{MR}$ and $K_1^{MR}\dot{h}^{MR}$ arise from the safety function dynamics in \eqref{eqn:ddot_h_MR} and \eqref{eqn:hocbf_cnd}. By scaling this worst-case decrease over the actuator response interval $1/p_1$, the tightening parameter $\Delta^{MR}(t)$ provides the minimum safety buffer required to maintain feasibility of the safety constraint despite bounded control authority
\end{remark}
We now introduce a lateral acceleration component for the attacker that enforces safety under the conservative formulation, given by
\begin{align}
    {a_A^b}^{MR} = \frac{
       \sum_{i\in\mathcal{D}} w_i^{MR} \left(\frac{(\dot{r}_{Ai})^{\,2}}{\beta}+a_i\sin\sigma_{iA} - r_{Ai}\dot{\theta}_{Ai}^2\right)- \frac{(\dot{h}^{MR})^{2}}{\beta}       -K_1^{MR}\dot{h}^{MR}-K_2^{MR}\left({\psi_1}^{MR}-\Delta^{MR}\right)}{\sum_{i\in\mathcal{D}} w_i^{MR}\sin\sigma_{Ai}}, \label{eqn:ab_mr}
\end{align}
where $K_1^{MR}>0$ and $K_2^{MR}>0$ are controller gains. The above lateral acceleration component \eqref{eqn:ab_mr} is obtained by imposing the higher-order control barrier function condition,
\begin{align}
    \dot{\psi}_1^{MR}\left(x\right) + \alpha_{2}\left( {\psi}_1^{MR}-\Delta^{MR}\right) \geq0, \label{eqn:safety_conser}
\end{align}
which ensures that the set $\mathcal{C}_1^{MR}$ is forward-invariant, allowing the attacker to remain outside the defenders' maximum range.

The target-interception component remains unchanged from the EZ-based safety formulation in \Cref{sec:our_app}. Using \eqref{eqn:inter_law} and \eqref{eqn:ab_mr}, the overall desired lateral acceleration is constructed as a convex combination of the target-interception and maximum-range avoidance components, given by
\begin{align}
    {a_A^d}^{MR}=\alpha^{MR}{a_A^b}^{MR} + \left(1-\alpha\right)a_A^T, \label{eqn:des_ctrl_cons}
\end{align}
where $\alpha^{MR}$ denotes a scalar continuous and monotonically decreasing switching function 
\begin{align}
    \alpha(\psi)^{MR} = \begin{cases}
        1, &\text{if}\;\psi^{MR}(a_A^T)\leq0,\\
        0, &\text{if}\;\psi^{MR}(a_A^T)\geq\varepsilon
    \end{cases}. \label{eqn:switch_mr}
\end{align}

\begin{remark}Similar to \Cref{lem:safety}, it can be shown that the above switching condition leads to the boundary
\begin{align}
    \psi(a_A^T)^{MR} = \frac{\dot{h}_c^{\,2}}{\beta} + \sum_{i\in\mathcal{D}} w_i \left(-\frac{\dot{r}_{Ai}^{\,2}}{\beta}-a_i\sin\sigma_{iA} + r_{Ai}\dot{\theta}_{Ai}^2+ \sin\sigma_{Ai} \right)a_A^T+K_1^{MR}\dot{h}^{MR}+K_2^{MR}\left({\psi_1}^{MR}-\Delta_{MR}\right) =0, \label{eqn:cons_cond}
\end{align}
such that ensuring $\psi(a_A^T)^{MR}\geq0$ is equivalent to ensuring safety, that is, $\psi(a_A^d)\geq0$. This boundary characterizes the influence of the relative attacker-defender motion on the behavior of the conservative safety function. In particular, the first three terms in the above boundary expression denote $\dot{h}(a_A^T)^{MR}$, while the last two terms are positive constants to ensure the control barrier function condition \eqref{eqn:safety_conser}.
\end{remark}
To account for input constraints using the saturation model in \eqref{eqn:sat_model}, we define the tracking error between the actual and the desired acceleration as
\begin{align}
    z^{MR} = a_A^{MR}-{a_A^d}^{MR},
\end{align}
which, in a way, represents the mismatch between the inner and outer control loops.
Accordingly, the commanded lateral acceleration for the attacker is designed to achieve target interception while ensuring that it remains outside the defenders’ maximum engagement range, and is given by
\begin{align}
    {a_A^c}^{MR}
    =
    \frac{
    p_1 a_A^{MR}
    + \left(\dot{a}_A^{d}\right)^{MR}
    + \left(1-\alpha^{MR}\right)\dfrac{\sigma_{AT}}{v_A}
    + \alpha^{MR}\psi_1^{MR}g_h^{MR}
    - K_a^{MR}\left(z^{MR}\right)
    }{
    1-\left(\dfrac{a_A^{MR}}{a_{\max}}\right)^n
    }, \label{eqn:ctrl_cmd_mr}
\end{align}
where $K_a^{MR}>0$ is a gain and $\Delta^{MR}(t)$ denotes the time-varying tightening term introduced to compensate for the transient effect of the bounded actuator dynamics. The first two terms in \eqref{eqn:ctrl_cmd_mr} compensate for the actuator dynamics and the time variation of the desired acceleration. The next two terms in \eqref{eqn:ctrl_cmd_mr} arise from the recursive backstepping design corresponding to the interception and conservative safety objectives. The last term in \eqref{eqn:ctrl_cmd_mr} is a robust corrective action that drives the acceleration tracking error to zero.

\begin{theorem}
    Consider the target–attacker–defender relative kinematics \eqref{eqn:rel_dyn_1}–\eqref{eqn:rel_dyn_2} with the input saturation model \eqref{eqn:sat_model}. Suppose the attacker’s commanded lateral acceleration is given by \eqref{eqn:ctrl_cmd_mr}, with the resulting actual lateral acceleration governed by \eqref{eqn:des_ctrl_cons}. Then the attacker intercepts the target while remaining outside the defenders' conservative maximum range boundaries and respecting the actuator bounds, provided the controller gains satisfy
     \begin{align}
     K_I>&\sup_{\alpha^{MR}\in(0,1)}\frac{2v_A+\epsilon_\sigma^{MR} + L_\alpha^{MR}}{
            2\left(R_i+c_i\right)(1-\alpha^{MR})^2} + w_2,\\ K_2^{MR}>&\sup_{\alpha^{MR}\in(0,1)}\frac{\epsilon_h^{MR}+\epsilon_\Delta^{MR}+L_\alpha^{MR}}{2\left(\alpha^{MR}\right)^2} +w_2, \label{eqn:suff_cond_MR}
    \end{align}
    $K_1^{MR}>0$ and $K_a>0$, where $\epsilon_\sigma^{MR}>0, \epsilon_h^{MR}>0, \epsilon_\Delta^{MR}>0$ and $ w_2>0$ are constants. \label{thm:thm_2}
\end{theorem}
\begin{proof}
    Consider the Lyapunov function candidate associated with the target-interception and conservative-safety objectives,
\begin{align}
V_1^{MR}
=
\frac{1}{2}\left(1-\alpha^{MR}\right)\sigma_{AT}^2
+
\frac{1}{2}\alpha^{MR}\left(\psi_1^{MR}-\Delta^{MR}\right)^2 +\frac{1}{2}z^{MR}.
\end{align}
Differentiating $V_1^{MR}$ with respect to time and using \eqref{eqn:sat_model}, \eqref{eqn:ddot_h_MR}, \eqref{eqn:hocbf_cnd} and the relationship $a_A^{MR}=z^{MR}+\left(a_A^d\right)^{MR}$, we obtain the dynamics of the Lyapunov candidate $V_1^{MR}$ along the closed-loop trajectories as
\begin{align}
\dot{V}_1^{MR}
=&
(1-\alpha^{MR})\sigma_{AT}\dot{\sigma}_{AT}
+
\alpha^{MR}\left(\psi_1^{MR}-\Delta^{MR}\right)\left(\ddot{h}^{MR}+K_1^{MR}\dot{h}^{MR}\right)-\alpha^{MR}\left(\psi_1^{MR}-\Delta^{MR}\right)\dot{\Delta}^{MR}
\nonumber \\&+
\frac{\dot{\alpha}^{MR}}{2}\left(
\left(\psi_1^{MR}\right)^2-\sigma_{AT}^2\right)+ z^{MR}\dot{z}^{MR}\nonumber\\
=&
(1-\alpha^{MR})\sigma_{AT}\left(\frac{a_A^d}{v_A}+\frac{v_A}{r_{AT}}\sin\sigma_{AT}\right) 
+
\alpha^{MR}\left(\psi_1^{MR}-\Delta^{MR}\right)\left(f_h^{MR}+g_h^{MR}a_A^d+K_1^{MR}\dot{h}^{MR}\right)\nonumber\\&-\alpha^{MR}\left(\psi_1^{MR}-\Delta^{MR}\right)\dot{\Delta}^{MR} +
\frac{\dot{\alpha}^{MR}}{2}\left(
\left(\psi_1^{MR}\right)^2-\sigma_{AT}^2\right)
\nonumber\\&+\left[\left(1-\alpha^{MR}\right)\frac{\sigma_{AT}}{v_A}+\alpha^{MR}\psi_1^{MR}g_h^{MR}+\left[1-\left(\frac{a_A}{a_{\max}}\right)^n\right]\left(a_A^c\right)^{MR} - p_1 a_A^{MR}-\left(\dot{a}_A^d\right)^{MR}\right]z^{MR}.
\end{align}
Now, using the proposed desired lateral acceleration \eqref{eqn:des_ctrl_cons} and the proposed commanded lateral acceleration \eqref{eqn:ctrl_cmd_mr}, we can simplify the above equation to
\begin{align}
    \dot{V}_1^{MR}=& -\left(1-\alpha^{MR}\right)^2K_I\sigma_{AT}^2 -\left(\alpha^{MR}\right)^2K_2^{MR}\left(\psi_1^{MR}\right)^2-K_a^{MR}\left(z^{MR}\right)^2 -\alpha^{MR}\left(\psi_1^{MR}-\Delta^{MR}\right)\dot{\Delta}^{MR}\nonumber\\
    &+\alpha(1-\alpha)\Bigg[\sigma_{AT}\left(a_A^b+\frac{v_A}{r_{AT}}\sin\sigma_{AT}\right)+\psi_1^{MR}\left(f_h^{MR}+g_h^{MR}a_A^T\right) \Bigg]+
\frac{\dot{\alpha}^{MR}}{2}\left(
\left(\psi_1^{MR}\right)^2-\sigma_{AT}^2\right). \label{eqn:lyap_dot_mr}
\end{align}
We now analyze the dynamics of the Lyapunov function candidate in the three regions of the state space, according to the switching condition \eqref{eqn:switch_mr}. In \textbf{Region I (pure target interception mode)}, $\alpha^{MR}=0$ and $\psi^{MR}(a_A^T)\geq\epsilon$, which implies the target interception term $a_A^T$ is purely activated in $\left(a_A^d\right)^{MR}$. Under such a condition $\dot{\alpha}^{MR}=0$, and $\alpha^{MR}\left(1-\alpha^{MR}\right)=0$, which renders
    \begin{align}
        \dot{V}_1^{MR} = -K_I\sigma_{AT}^2 - K_a^{MR} \left(z^{MR}\right)^2 <0, \;\forall \;(\sigma_{AT}, z^{MR}) \in \mathbb{R}^2 \setminus(0,0),
    \end{align}
    if the controller gains are selected as $K_I>0$ and $K_a^{MR}>0$. This implies that both $\sigma_{AT}$ and $z^{MR}$ converge asymptotically to zero. This aligns the attacker's heading angle along the LOS to the target, yielding $\dot{r}_{AT}=-v_A<0$ and guaranteeing monotonic range. Also, the attacker will perfectly track the desired acceleration, that is, $a_A={a_A^d}^{MR}=a_A^T$, since $\lim_{t\to\infty}z^{MR}\to 0$. 
    
    In \textbf{Region II (pure maximum-engagement-range avoidance mode)}, $\psi^{MR}(a_A^T)\leq0$ and $\alpha^{MR}=1$, resulting in the safety component $\left(a_A^b\right)^{MR}$ to be fully active in the desired acceleration $\left(a_A^d\right)^{MR}$. Such a condition renders 
    \begin{align}
        \dot{V}_1^{MR} =& -K_2^{MR}\left(\psi_1^{MR}-\Delta^{MR}\right)^2 - K_a^{MR} \left(z^{MR}\right)^2-\alpha^{MR}\left(\psi_1^{MR}-\Delta^{MR}\right)\dot{\Delta}^{MR}\nonumber\\\leq&-\frac{K_2^{MR}}{2}\left(\psi_1^{MR}-\Delta^{MR}\right)^2-K_a^{MR}\left(z^{MR}\right)^2+\frac{L_{\Delta^{MR}}}{2K_2^{MR}}
    \end{align}
    since $\dot{\alpha}^{MR}=0$ and $\alpha^{MR}\left(1-\alpha^{MR}\right)=0$. It follows from the above expression that if $K_2^{MR}>0$ and $K_a>0$, then $\dot{V}_1^{MR}<0$ outside the compact set $$\Omega_{II}^{MR} \coloneqq \left\{ (\psi_1^{MR}-\Delta^{MR}, z^{MR})\in \mathbb{R}^2\middle|\dfrac{K_2^{MR}}{2}\left(\psi_1^{MR}-\Delta\right)^2+K_a^{MR}\left(z^{MR}\right)^2\leq\dfrac{\left(L_\Delta^{MR}\right)^2}{2K_2^{MR}}\right\}.$$ Therefore, $\psi_1^{MR}-\Delta$ and $z^{MR}$ are uniformly ultimately bounded in {Region II} with the ultimate performance bounds given as $$|\psi_1^{MR}-\Delta^{MR}|\leq\dfrac{L_\Delta^{MR}}{K_2^{MR}},~~|z^{MR}|\leq\sqrt{\dfrac{\left(L_\Delta^{MR}\right)^2}{2K_2^{MR}K_a^{MR}}}.$$ This results in $\psi_1^{MR}\to\Delta^{MR}+\epsilon_1^{MR}$ and $z\to \epsilon_2^{MR}$, where the residual errors $\epsilon_1$ and $\epsilon_2$ are bounded by the ultimate performance limits as derived above and can be made arbitrarily small. Since $\psi_1^{MR}=\dot{h}^{MR}+K_1^{MR}h^{MR}$, we have $\dot{h}^{MR}+K_1^{MR}h^{MR}\to\Delta^{MR}+\epsilon_1$. In steady state, when $\dot{h}^{MR}\to0$, $h^{MR}$ converges to a non-negative value, that is, $h^{MR}\to\dfrac{\Delta^{MR}+\epsilon_1^{MR}}{K_1^{MR}}\geq0$, ensuring that the safety is maintained asymptotically. Therefore, this implies that in {Region II}, the guidance law prioritizes safety and steers the attacker toward the boundary of the conservative safe set $\mathcal{C}_1^{MR}$. Additionally, $z^{MR}\to0$ asymptotically, resulting in the attacker being able to perfectly track the desired lateral acceleration, that is, $a_A\to a_A^d$.

    In \textbf{Region III (blended mode)}, $0<\psi^{MR}(a_A^T)<\epsilon^{MR}$, $0<\alpha^{MR}<1$, such that a convex combination of $a_A^b$ and $a_A^T$ is active in $a_A^d$, and  simplifies \eqref{eqn:lyap_dot_mr} to
    \begin{align}
       \dot{V}_1^{MR} \leq & -\left[(1-\alpha)^2K_I-\frac{1}{4}\left(\frac{v_A}{r_{AT}}+ \frac{\epsilon_\sigma^{MR}}{2}\right)-\frac{L_\alpha^{MR}}{2}\right]\sigma_{AT}^2 -\left[\alpha^2K_2^{MR} - \frac{\epsilon_h^{MR} +\epsilon_\Delta^{MR}}{8}+ \frac{L_\alpha^{MR}}{2}\right]\left(\psi_1^{MR}-\Delta^{MR}\right)^2 \nonumber\\&-K_a^{MR}\left(z^{MR}\right)^2 + \frac{\left|\left(a_A^b\right)^{MR}\right|^2}{8\epsilon_\sigma} + \frac{\left|\phi^{MR}\right|^2}{8\epsilon_h^{MR}} + \frac{L_\Delta^{MR}}{2\epsilon_\Delta^{MR}},
       \end{align}
      where $\phi^{MR}=f_h^{MR}+g_h^{MR}a_A^T$. The above expression provides the sufficient condition as presented in \eqref{eqn:suff_cond_MR} on the gains (after substituting $r_{AT}\geq R_i+c_i$, since in Region III, interception cannot occur and therefore $r_{AT}$ is bounded away from zero) to ensure that the first two terms in the above expression are always negative definite, yielding
    \begin{align}
        \dot{V}_1^{MR} = & -w_2\sigma_{AT}^2-w_2\left(\psi_1^{MR}-\Delta^{MR}\right)^2 -K_a^{MR}\left(z^{MR}\right)^2 + \overline{d}^{MR}
    \end{align}
    where $$\overline{d}^{MR}= \frac{\left|\left(a_A^b\right)^{MR}\right|^2}{8\epsilon_\sigma} + \frac{|\phi|^2 }{8\epsilon_h^{MR}}+ \frac{L_\Delta^{MR}}{2\epsilon_\Delta^{MR}}$$ denotes the upper bound on the residual terms. It follows from the above expression that the $\dot{V}_1^{MR}$ decreases outside the compact set $$\Omega_{III}^{MR}\coloneqq\Big\{(\sigma_{AT}, \psi_1^{MR},z^{MR})\in \mathbb{R}^3\big|w_2\sigma_{AT}^2 +w_2\left(\psi_1^{MR}-\Delta^{MR}\right)^2 + K_a^{MR}\left(z^{MR}\right)^2\leq\overline{d}^{MR}\Big\}$$ with ultimate performance bounds, $|\sigma_{AT}|, \left|\psi_1^{MR}-\Delta^{MR}\right|\leq \sqrt{\dfrac{\overline{d}^{MR}}{w_2}}$ and $\left|z^{MR}\right|\leq \sqrt{\dfrac{\overline{d}^{MR}}{K_a^{MR}}}$. This implies that in {Region III}, $\sigma_{AT}\to \epsilon_3^{MR}$, $\psi_1^{MR}\to \Delta^{MR}+ \epsilon_4^{MR}$ and $z\to\epsilon_5^{MR}$, where $\epsilon_3^{MR}, \epsilon_4^{MR}$ and $\epsilon_5^{MR}$ are positive constants. Consequently, the attacker maintains a bounded LOS misalignment, $\sigma_{AT} \not\to 0$, within a bounded neighborhood of the tightened safety boundary, while its lateral acceleration tracks the desired acceleration up to a bounded residual error, $a_A\to {a_A^d}^{MR} +\epsilon_5$. Therefore, the attacker remains outside the maximum engagement range of all the defenders, preserving safety at all times. 
    
    When the attacker's state lies in Region I, it is away from the defender's threat, and the attacker prioritizes target interception. While in Regions II and III, the attacker prioritizes safety.  In addition, the smooth saturation dynamics in \eqref{eqn:sat_model} ensures that the actual control input $a_A$ remains within the admissible limits for all time. Therefore, the proposed guidance law guarantees target interception, practical safety from the defenders, while respecting the constraints on the attacker's control inputs. This completes the proof. 
\end{proof}
\begin{remark}
    It is essential to note that in comparison to EZ-based design, the maximum engagement range-based design relies on the knowledge of the defender's control actions (see \eqref{eqn:ab_mr} and \eqref{eqn:cons_cond}). This dependency introduces a significant information requirement and limits the applicability of range-based safety formulations in scenarios where defender strategies are unknown, uncertain, or dynamically changing.
\end{remark}
Since the conservative safety function has a relative degree of two with respect to the attacker’s lateral acceleration, the resulting design relies on a higher-order control barrier function formulation, which requires additional modeling information and a more complex controller structure. In contrast, the earlier proposed EZ-based approach leads to a simpler first-order control barrier function design. While both approaches guarantee the safety of the attacker under bounded control inputs, they differ in their operational behavior. In the maximum engagement-range-based formulation, the attacker maintains a conservative distance and never enters the defenders’ engagement range. On the other hand, the EZ-based strategy allows the attacker to penetrate the defenders’ engagement range intelligently while ensuring the attacker cannot be neutralized by remaining outside the EZ. Thus, the maximum engagement-range-based formulation provides a conservative safety baseline, whereas the EZ-based design enables more aggressive yet provably safe interception strategies.

\section{Simulation results}
We now present simulation results demonstrating the efficacy of the proposed EZ-based and maximum engagement-range-based methods across varying numbers of defenders and defense strategies. For all simulation results, the attacker moves at a constant speed $v_A = 1\,\mathrm{m/s}$, and the controller gains for the EZ-based design are selected as $K_I = 0.7$ for target interception, $K_2 = 0.9$ for EZ avoidance, and $ K_a = 0.9$. The controller gains for the maximum engagement-range-based design are chosen as $K_1^{MR}=1.8$, $K_2^{MR}=1.5$, and $K_a^{MR}=15$.  Moreover, the safety aggregation parameters are selected as $\beta=0.5$ and $\beta^{MR}=0.9$ in the two cases. Each defender moves at a speed ratio of $\mu_i = 0.7$ relative to the attacker, has engagement range $R_i = 1.5\,\mathrm{m}$, and a capture radius $r_i = 0.5\,\text{m}$ resulting in maximum engagement radius of $2\mathrm{m}$. In the trajectory plots that follow, star-shaped markers denote the initial positions of the agents, while hollow circular markers represent their locations at intermediary times. Additionally, in the following figures, $\left(\cdot\right)^{EZ}$ represents variables associated with EZ-based design, while $\left(\cdot\right)^{MR}$ represents those for maximum engagement-range-based design. To better visualize the agents' trajectories, we also provide animations for all the following results at \tr{\href{https://youtu.be/5DsaCRfdgmc}{https://youtu.be/5DsaCRfdgmc}}.
\begin{figure}[h!]
	\centering
	\begin{subfigure}[t]{.49\linewidth}
    \centering
    \includegraphics[width=\linewidth]{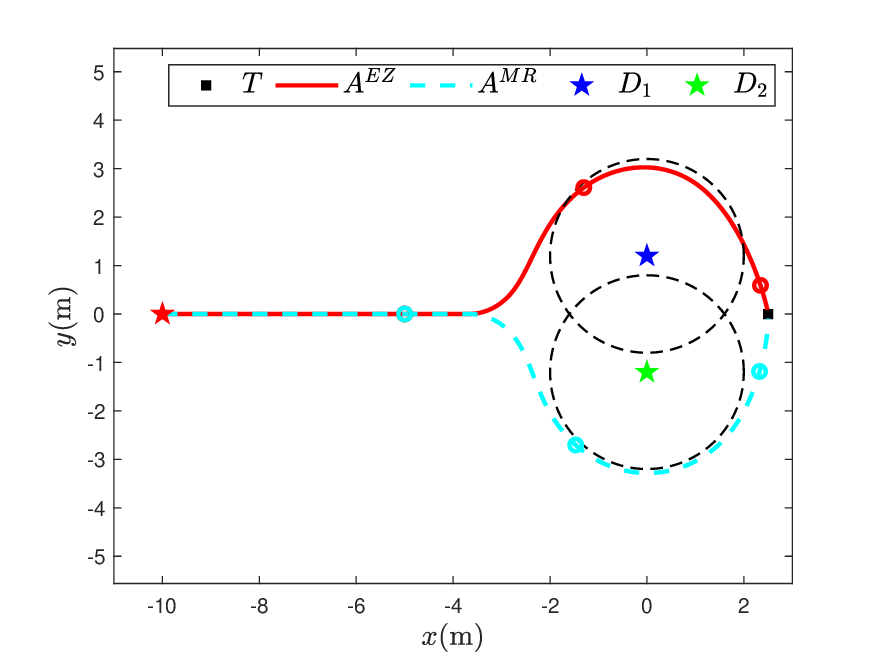}
    \caption{Attacker's trajectories.}
    \label{fig:def_2_traj}
    \end{subfigure}
	\begin{subfigure}[t]{0.49\linewidth}
		\centering
		\includegraphics[width=\textwidth]{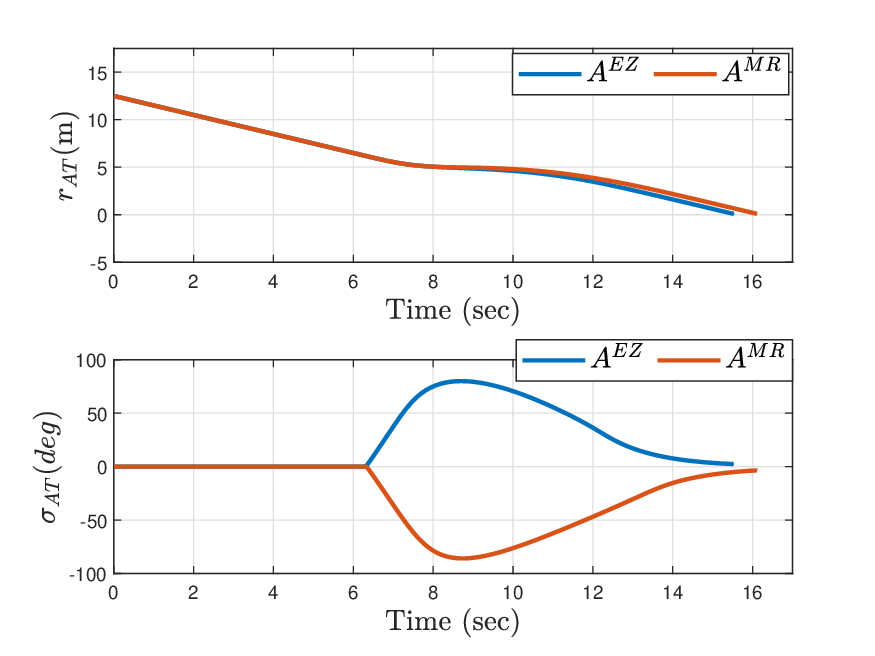}
		\caption{Attacker-to-target distance and bearing angle.}
		\label{fig:rel_ini_1}
	\end{subfigure}	
    \begin{subfigure}[t]{0.49\linewidth}
		\centering
		\includegraphics[width=\textwidth]{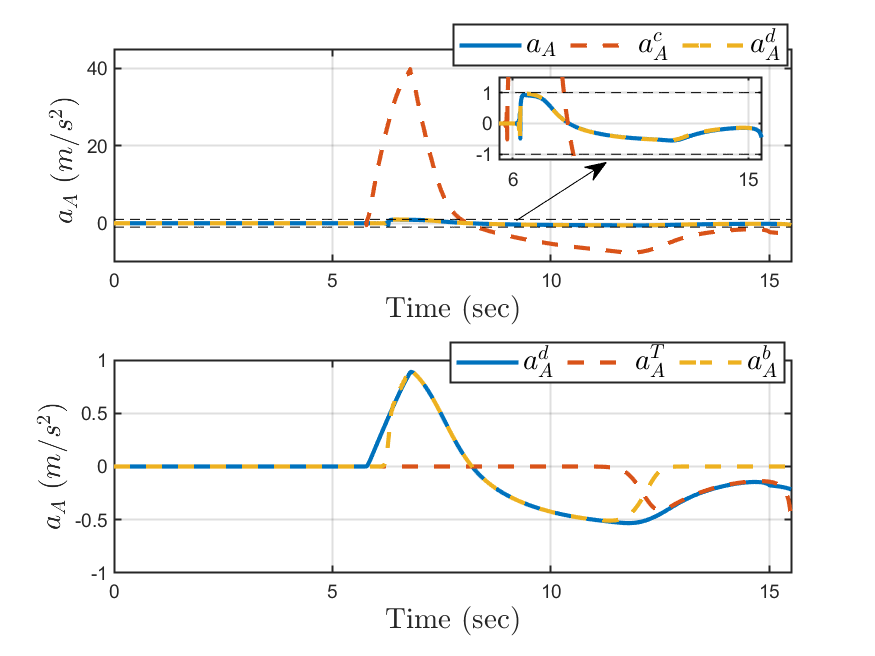}
		\caption{Control inputs (EZ formulation).}
		\label{fig:ctrl_saf_ini_1}
	\end{subfigure}
    \begin{subfigure}[t]{0.49\linewidth}
		\centering
		\includegraphics[width=\textwidth]{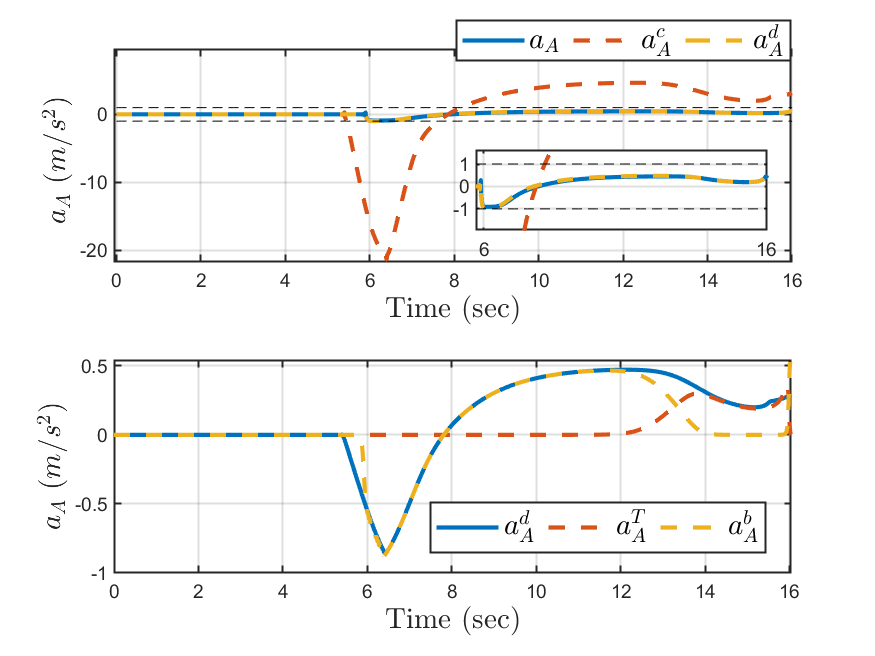}
		\caption{Control inputs (maximum range formulation).}
		\label{fig:ctrl_cons_ini_1}
	\end{subfigure}
    \begin{subfigure}[t]{0.49\linewidth}
		\centering
		\includegraphics[width=\textwidth]{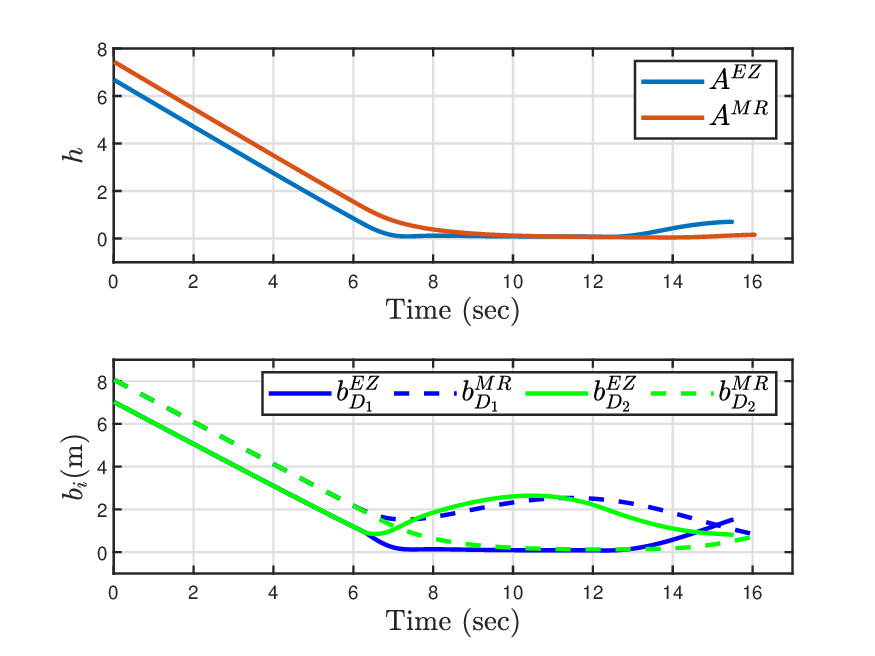}
		\caption{Evolution of safety functions.}
		\label{fig:saf_1}
	\end{subfigure}
    \begin{subfigure}[t]{0.49\linewidth}
		\centering
		\includegraphics[width=\textwidth]{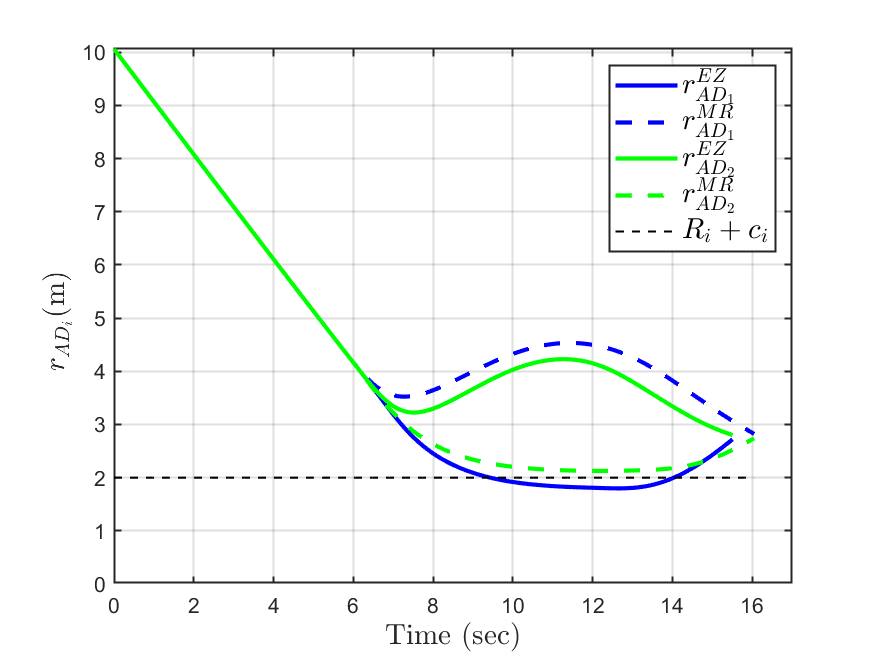}
		\caption{EZ penetration.}
		\label{fig:pen_eng_1}
	\end{subfigure}
	\caption{Target interception by avoiding two stationary defenders.}
	\label{fig:2def}
\end{figure}
\begin{figure}[h!]
	\centering
	\begin{subfigure}[t]{.49\linewidth}
    \centering
    \includegraphics[width=\linewidth]{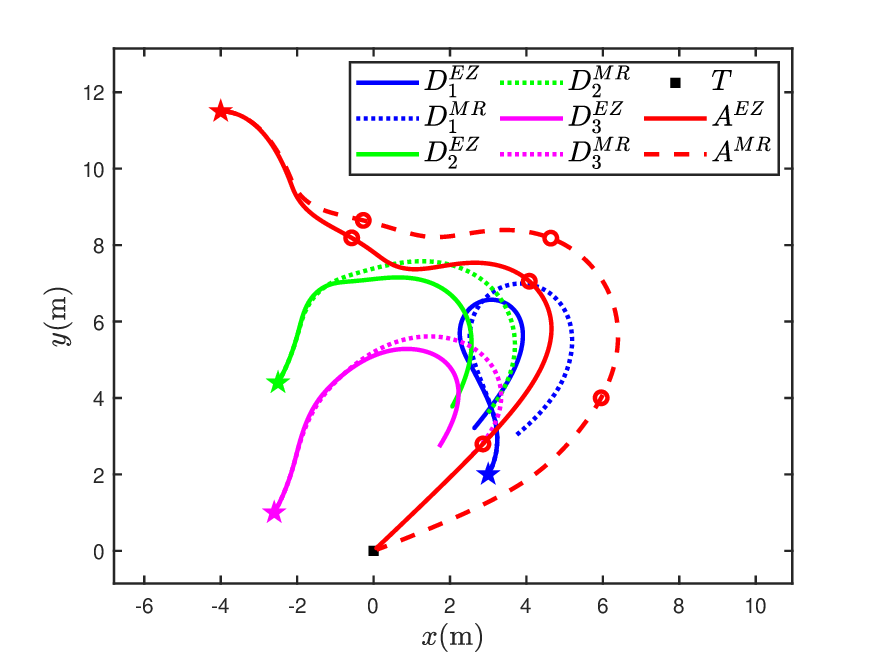}
    \caption{Attacker's trajectories.}
    \label{fig:safe_traj_ini_2}
    \end{subfigure}
    \begin{subfigure}[t]{0.49\linewidth}
		\centering
		\includegraphics[width=\textwidth]{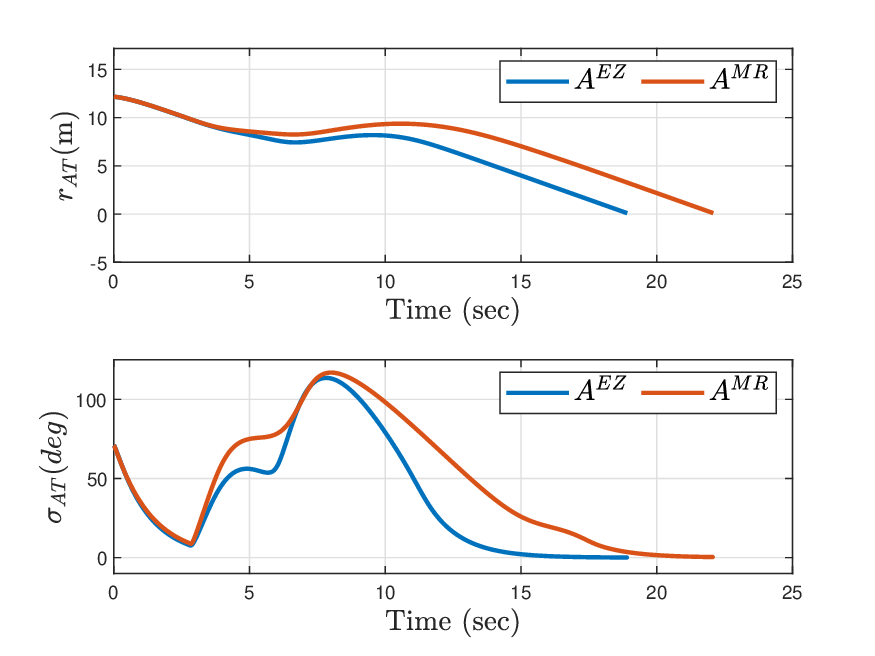}
		\caption{Attacker-to-target distance and bearing angle.}
		\label{fig:rel_ini_2}
	\end{subfigure}	
    \begin{subfigure}[t]{0.49\linewidth}
		\centering
		\includegraphics[width=\textwidth]{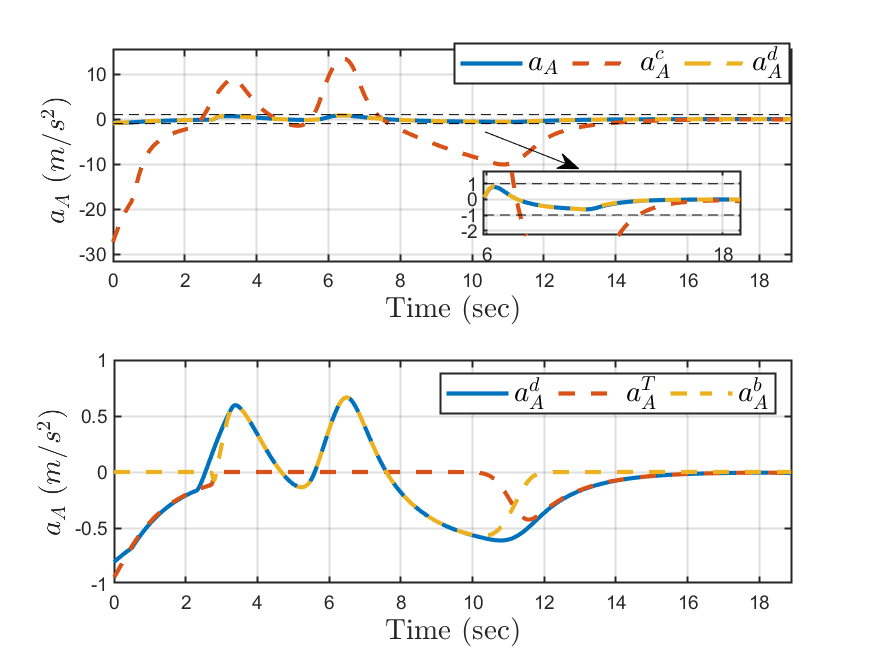}
		\caption{Control inputs (EZ formulation).}
		\label{fig:ctrl_saf_ini_2}
	\end{subfigure}
    \begin{subfigure}[t]{0.49\linewidth}
		\centering
		\includegraphics[width=\textwidth]{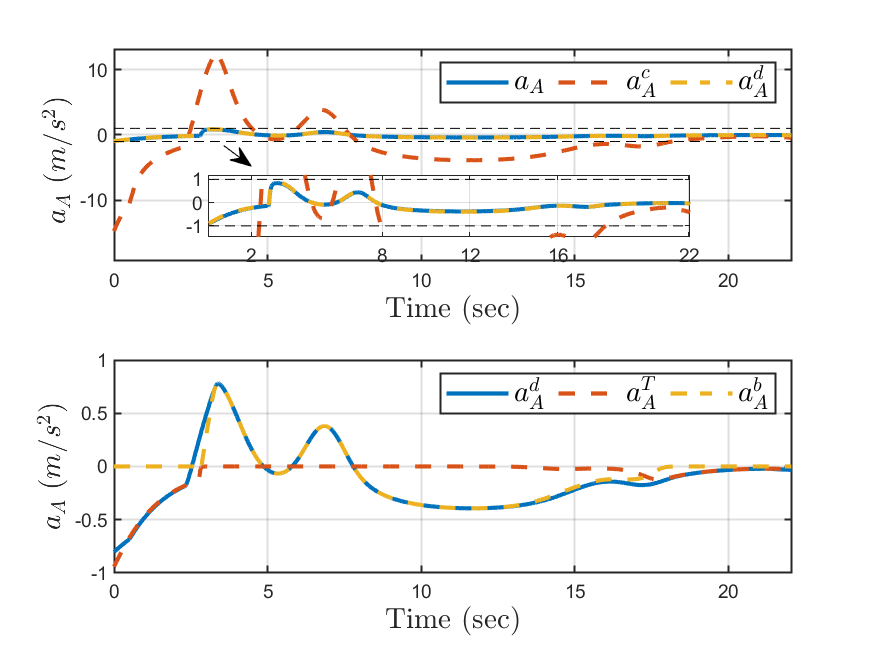}
		\caption{Control inputs (maximum range formulation).}
		\label{fig:ctrlcons_ini_2}
	\end{subfigure}
    \begin{subfigure}[t]{0.49\linewidth}
		\centering
		\includegraphics[width=\textwidth]{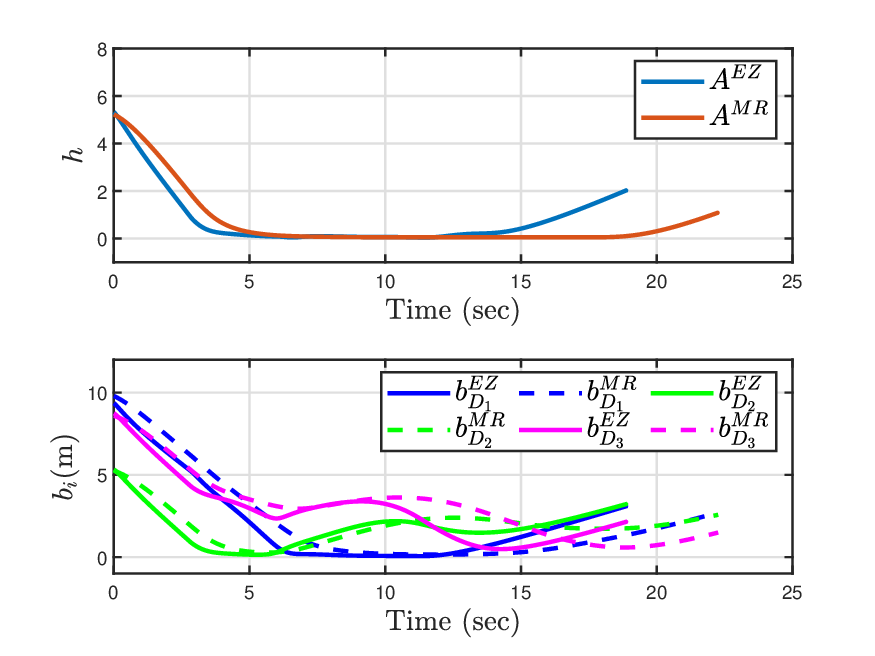}
		\caption{Evolution of safety functions.}
		\label{fig:h_saf_ini_2}
	\end{subfigure}
    \begin{subfigure}[t]{0.49\linewidth}
		\centering
		\includegraphics[width=\textwidth]{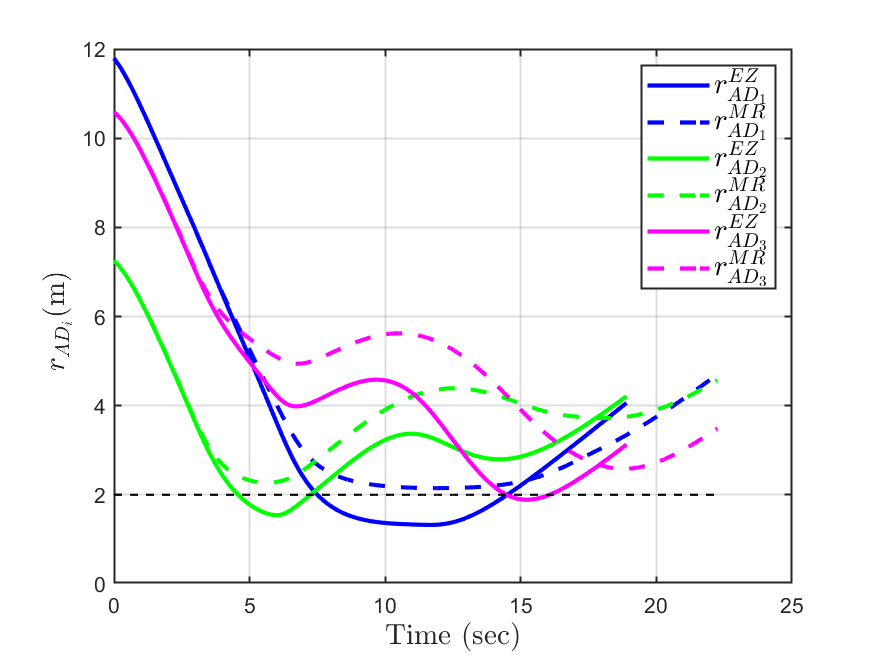}
		\caption{EZ penetration.}
		\label{fig:pen_eng_2_safe}
	\end{subfigure}
	\caption{Comparison of attacker performance with three moving defenders.}
	\label{fig:3def_saf}
\end{figure}
\begin{figure}[h!]
	\centering
	\begin{subfigure}[t]{.49\linewidth}
    \centering
    \includegraphics[width=\linewidth]{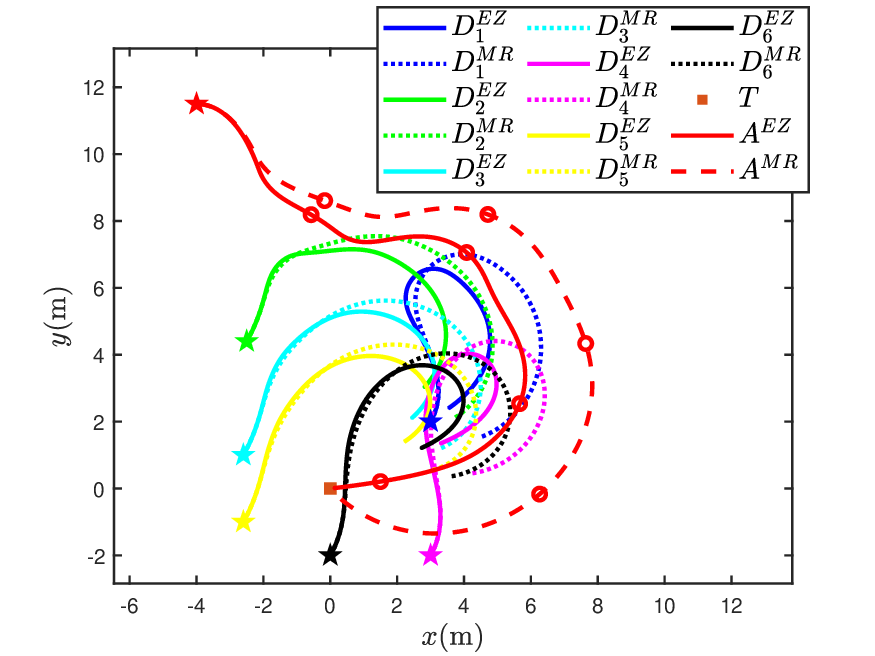}
    \caption{Attacker's trajectories.}
    \label{fig:safe_traj_ini_3}
    \end{subfigure}
    \begin{subfigure}[t]{0.49\linewidth}
		\centering
		\includegraphics[width=\textwidth]{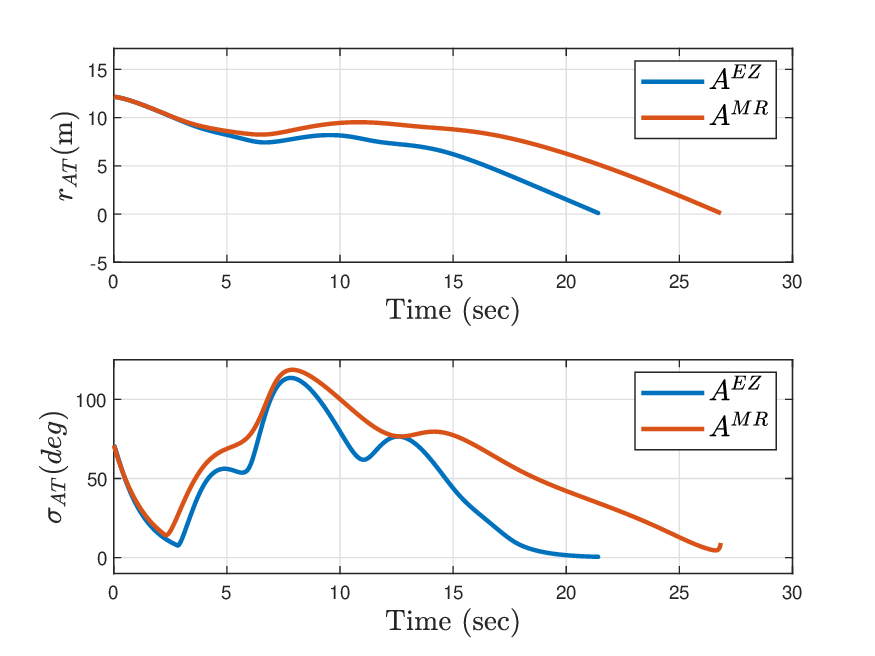}
		\caption{Attacker-to-target distance and bearing angle.}
		\label{fig:rel_ini_3}
	\end{subfigure}	
    \begin{subfigure}[t]{0.49\linewidth}
		\centering
		\includegraphics[width=\textwidth]{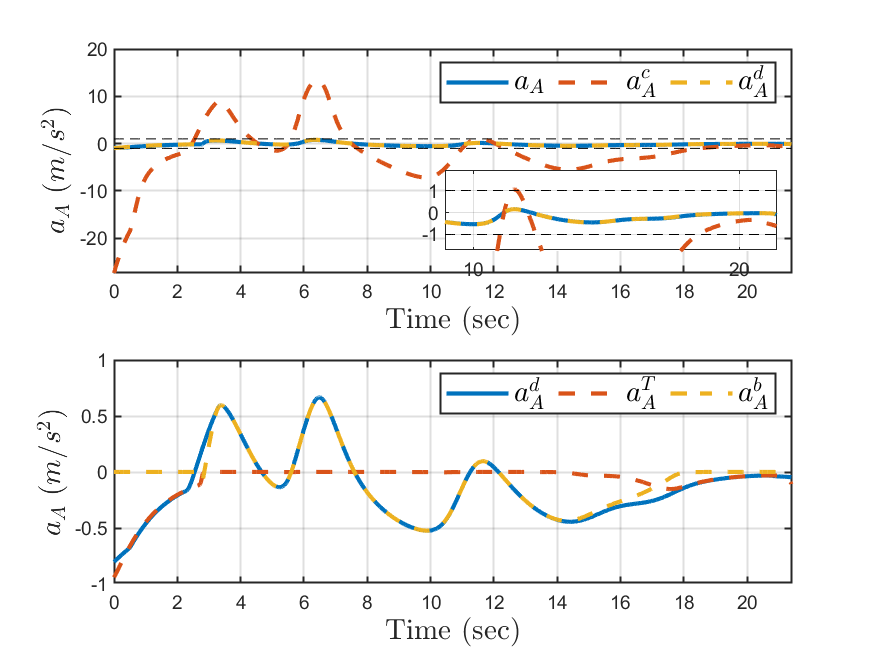}
		\caption{Control inputs (EZ formulation).}
		\label{fig:ctrl_saf_ini_3}
	\end{subfigure}
    \begin{subfigure}[t]{0.49\linewidth}
		\centering
		\includegraphics[width=\textwidth]{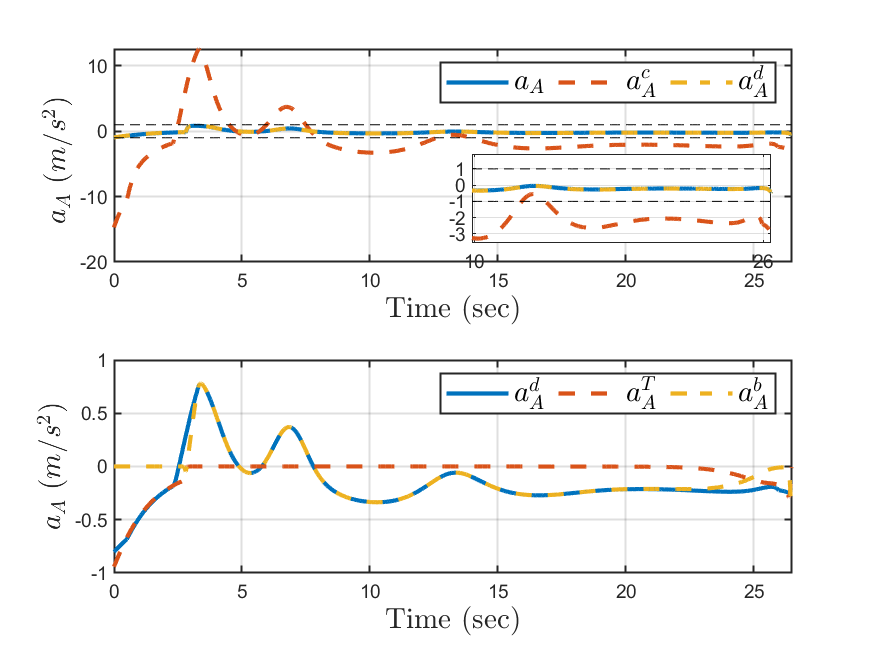}
		\caption{Control inputs (maximum range formulation).}
		\label{fig:ctrlcons_ini_3}
	\end{subfigure}
    \begin{subfigure}[t]{0.49\linewidth}
		\centering
		\includegraphics[width=\textwidth]{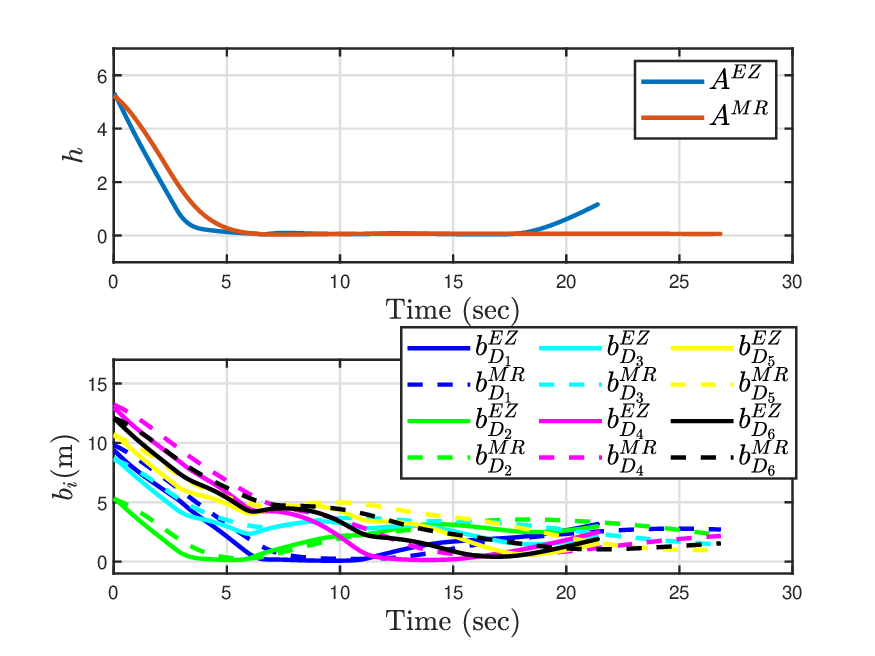}
		\caption{Evolution of safety functions.}
		\label{fig:h_saf_ini_3}
	\end{subfigure}
    \begin{subfigure}[t]{0.49\linewidth}
		\centering
		\includegraphics[width=\textwidth]{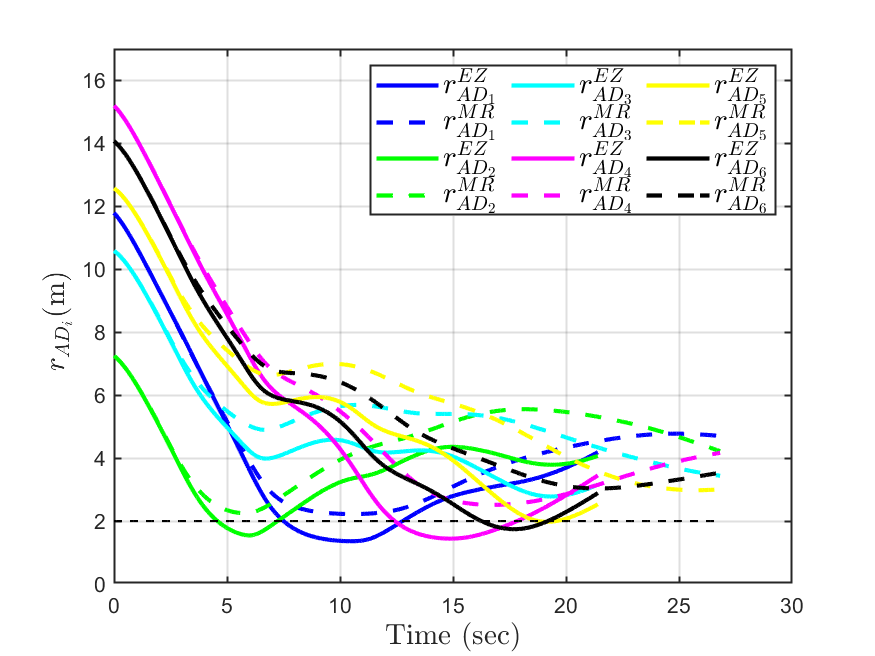}
		\caption{EZ penetration.}
		\label{fig:pen_eng_3_safe}
	\end{subfigure}
	\caption{Comparison of attacker performance with six moving defenders.}
	\label{fig:6def_saf}
\end{figure}
\begin{figure}[h!]
	\centering
	 \begin{subfigure}[t]{0.327\linewidth}
		\centering		\includegraphics[width=\textwidth]{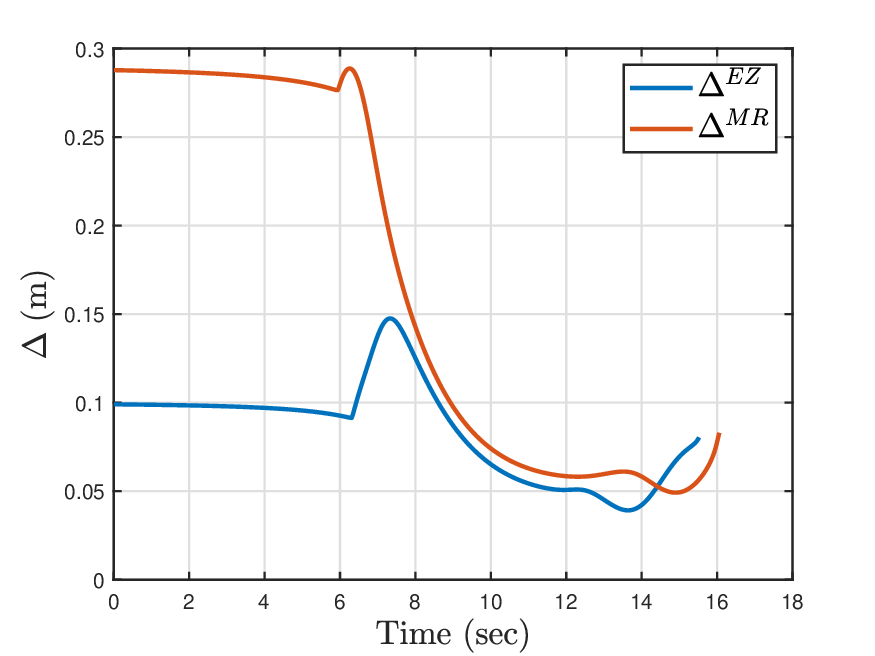}
		\caption{Two stationary defenders.}
		\label{fig:del_ini_1}
	\end{subfigure} \begin{subfigure}[t]{0.327\linewidth}
		\centering
		\includegraphics[width=\textwidth]{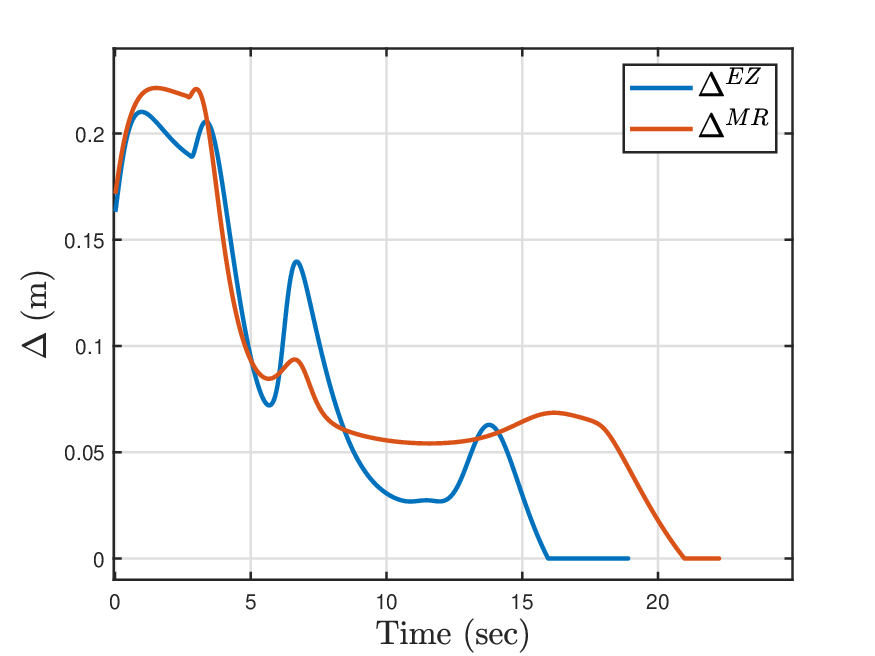}
		\caption{Three moving defenders.}
		\label{fig:del_ini_2}
	\end{subfigure}
    \begin{subfigure}[t]{0.327\linewidth}
		\centering
		\includegraphics[width=\textwidth]{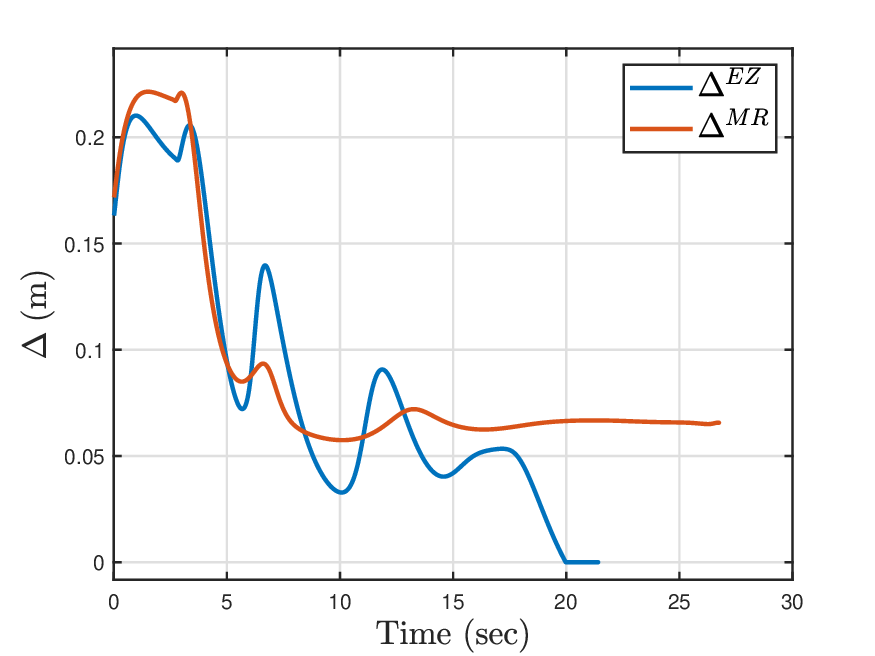}
		\caption{Six moving defenders.}
		\label{fig:del_ini_3}
	\end{subfigure}
	\caption{Comparison of the tightening parameter variation.}
	\label{fig:delcompare}
\end{figure}

\subsection{Target Interception avoiding Stationary Defenders}
In the first set of results (see \Cref{fig:2def}), we consider two stationary defenders, located at $([x_{D_1}, y_{D_1}]^\top = [0, 0.6]\mathrm{m}$, $[x_{D_2}, y_{D_2}]^\top = [0,-0.6]\mathrm{m}$. The attacker starts at  $([x_{A}, y_{A}]^\top = [0, -10]\mathrm{m}$ with a heading angle of $\gamma_A=0^\circ$, such that the initial LOS from attacker to the target intersects with the overlapping region between the engagement range of the defenders. The trajectory for both methods is compared in \Cref{fig:def_2_traj}, where the dotted circles centered at the respective defender's position represent the maximum engagement radius of the defender. Both approaches achieve successful target interception. However, the EZ-based approach enables the attacker to enter the defenders' maximum range, potentially allowing it to take a shorter path. In comparison, maximum engagement-range-based design forces the attacker to remain conservatively outside the defender's maximum effective range, resulting in a longer path to intercept the target. This observation is verified by comparing the target relative variables, as shown in \Cref{fig:rel_ini_1}. EZ-based design takes less time for the relative attacker-target range to converge to zero than the maximum engagement-range formulation. In general, the relative range to the target decreases over time and eventually converges to zero, indicating successful interception. The lead angle of the attacker relative to the target, denoted by $\sigma_{AT}$, also shows convergent behavior, that is, approaching zero nominally but temporarily deviating when the attacker executes avoidance maneuvers near the EZ. \Cref{fig:ctrl_saf_ini_1,fig:ctrl_cons_ini_1} presents the attacker’s control input profiles for both methods, illustrating the lateral acceleration converging to the desired lateral acceleration without violating the prescribed bounds. 
The third and fourth subplots in \Cref{fig:ctrl_saf_ini_1,fig:ctrl_cons_ini_1} compare the components of the desired lateral acceleration, demonstrating switching between target-seeking and safety-preserving modes and showing smooth transitions at different time instants depending on the perceived threat from nearby defenders.

\Cref{fig:saf_1} depicts the profiles of the safety constraint variables for both methods, where the aggregate safety function and the safety function remain strictly positive throughout the engagement, guaranteeing safety for the attacker at all times. However, it is worth mentioning that these variables have different physical interpretations across the two approaches. In the maximum engagement-range-based design, safety corresponds to strictly keeping the attacker outside each defender’s maximum engagement range, whereas EZ-based safety is defined with respect to an EZ, which captures the coupled geometric and dynamic interaction between the attacker and the defenders. Finally, we compare the relative distances between the attacker and the defenders in \Cref{fig:pen_eng_1}, where the dotted black line represents the defender's maximum range. The EZ-based design profile is shown in a solid line, while the maximum engagement-range-based design is shown in dotted lines. One can observe that the attacker's distance falls below the black dotted line at times when the attacker is within a defender's engagement range. In the maximum engagement range-based design, the relative attacker-defender distance never falls below the defender engagement range, clearly demonstrating the conservatism of such an approach.

\subsection{Target Interception avoiding Moving Defenders}
We now present results for moving defenders ($n=3$ and $n=6$), where the defenders adopt an aggressive stance and utilize pure pursuit to intercept the attacker by executing the lateral acceleration
\begin{align}
    a_{D_i} = -K_I v_i\left(\gamma_i-\theta_{iA}\right)+ \dfrac{v_iv_A \sin\left(\gamma_A-\theta_{iA}\right)}{r_{iA}} - \dfrac{v_i^2 \sin\left(\gamma_i-\theta_{iA}\right)}
{r_{iA}},\; \forall\;i\in \mathcal{D}, \label{eqn:def_ctrl}
\end{align}
where $\theta_{iA}$ denotes the LOS angle from the $i$\textsuperscript{th} defender to the attacker. In the results presented in \Cref{fig:3def_saf}, three defenders are initially positioned at $([x_{D_1}, y_{D_1}]^\top = [2.5,\, 2.5]\mathrm{m}$ , $[x_{D_2}, y_{D_2}]^\top = [-2,\, 4.4]\mathrm{m}$ , and $[x_{D_3}, y_{D_3}]^\top = [-2.9,\, 2.2]\mathrm{m}$, respectively. In the results presented in \Cref{fig:6def_saf}, first three defenders start as the same position as in the $n=3$ case, while the remaining defenders start at, $([x_{D_4}, y_{D_4}]^\top = [3,\, -2]\mathrm{m}$ , $[x_{D_5}, y_{D_5}]^\top = [-2.6,\, -1]\mathrm{m}$ , and $[x_{D_6}, y_{D_6}]^\top = [0,\, -2]\mathrm{m}$, respectively. All defenders are assumed to be homogeneous in their capabilities, similar to the stationary defender case. 

\Cref{fig:safe_traj_ini_2,fig:safe_traj_ini_3} illustrate the trajectory for the defenders and the attacker for both approaches, where the defenders can be seen to be moving closer to intercept the target following the control law presented in \eqref{eqn:def_ctrl}. One can observe that, following the EZ-based approach, the attacker takes a shorter path to intercept the target by penetrating the defender's engagement range, while the maximum engagement-range-based approach takes a longer route to comply with the strict requirement of fully avoiding a defender's range. The relative variables are compared in \Cref{fig:rel_ini_2,fig:rel_ini_3}, which shows the relative range and the bearing angles eventually converging to zero, with the convergence time difference between the methods increasing in the moving case compared to the stationary defenders scenario. However, between $5-12 \mathrm{sec}$, the relative variable profiles can be seen temporarily moving away from zero, corresponding to times when the attacker executes an evasive maneuver to avoid an EZ or a maximum engagement range. 

The safety constraints profiles are depicted in \Cref{fig:h_saf_ini_2,fig:h_saf_ini_3}. For both methods and both scenarios, the safety functions remain strictly positive, ensuring the attacker's safety against the defender's threat. The subplots \Cref{fig:ctrl_saf_ini_2,fig:ctrlcons_ini_2,fig:ctrl_saf_ini_3,fig:ctrlcons_ini_3} compare the different components of the attacker's lateral acceleration in the moving defender scenario. These profiles illustrate switching between the target-seeking and safety-preserving modes, with the control input dynamically adapting to the defenders' time-varying positions. Smooth transitions between these modes can be observed at different time instants as the attacker adjusts its maneuvering strategy in response to the evolving threat posed by the moving defenders. The comparison of the attacker-defender distance is presented in \Cref{fig:pen_eng_2_safe,fig:pen_eng_3_safe}, which demonstrates that the engagement-range-based method may temporarily lead the attacker into the defenders' engagement range without being unsafe, while the maximum engagement-range-based approach strictly prevents such an entry by forcing a conservative distance constraint.

Finally, we compare the tightening parameter profiles for both methods across all scenarios, as illustrated in \Cref{fig:delcompare}. It is observed that the tightening parameter assumes relatively larger values during the initial phase of the engagement and gradually decreases as the engagement progresses. This happens since during the initial phases, the attacker encounters the unsafe regions head-on and demands higher values of the tightening parameter. As the engagement proceeds, the attacker performs evasive maneuvers and progressively avoids the defenders’ threatening regions. Consequently, the attacker eventually aligns with the target along a relatively unobstructed LOS. Under such conditions, the safety constraint becomes less restrictive, and the required tightening margin decreases accordingly. As a result, the tightening parameter gradually converges toward zero at the end of the engagement, indicating that no additional safety margin is required to maintain safe operation. Furthermore, the tightening parameter remains strictly positive and evolves smoothly throughout the entire engagement for all scenarios, which ensures continuous enforcement of the safety constraints without introducing discontinuous variations in the control behavior.

\begin{table}[h!]
\centering
\caption{Comparison of target interception time.}
\label{tab:time_compare}
\begin{tabular}{|l|c|c|c|}
\hline
\textbf{Simulation Scenario} & \textbf{Proposed Method } & \textbf{Conservative safety constraint}&
\textbf{Noramlized Savings ($\%$)}\\
\hline
Static defenders ($n=2$) & 15.52 s & 15.97 s& 2.82\\ \hline
Moving defenders ($n=3$) & 18.91 s & 22.08 s&14.36 \\ \hline
Moving defenders ($n=6$) & 21.42 s & 26.82 s& 20.13 \\
\hline
\end{tabular}
\end{table}

\Cref{tab:time_compare} compares the interception times achieved by the EZ-based and maximum engagement-range-based formulations for defenders. The EZ-based approach yields a modest reduction in interception time for static defenders (2.82\%) and a substantially larger reduction for moving defenders (14.36\% and 20.13\%), compared to the maximum engagement-range formulation. The relatively smaller improvement in the stationary defender case arises because the defenders do not actively maneuver, leading to a largely static threat region in which both methods produce similar avoidance behavior. In contrast, when defenders are moving, the threat geometry evolves dynamically, which makes conservative distance-based constraints more restrictive. Under such conditions, the EZ-aware formulation more effectively captures the coupled kinematic and dynamic interaction between the attacker and defenders, allowing the attacker to exploit favorable engagement configurations while still maintaining safety guarantees. Consequently, the attacker can maneuver more efficiently around the defenders and maintain a more direct path toward the target, resulting in a shorter interception time. These results highlight the advantage of incorporating EZ-aware safety constraints, particularly in dynamic multi-defender scenarios where threat regions evolve over time.

\section{Conclusions}
In this work, we designed nonlinear guidance laws for an attacker to safely intercept a target while avoiding capture by the defenders and respecting the physical bounds on the attacker's lateral acceleration. In the first approach, the defender-induced regions where attacker interception by defenders is guaranteed were modeled as Engagement Zones (EZs) and directly incorporated into the guidance design. In the second approach, we design the safety constraints to ensure that the attacker remains outside the maximum engagement range of the defenders.
To ensure the designs respect the control input bounds, we also incorporate a smooth saturation model and introduce a tightening parameter to shrink the respective safe set and ensure permissible control input values. Further, a smooth minimum (log-sum-exp) function was adopted to aggregate risks across multiple zones into a unified safety measure. Stability analysis under a continuous switching function established safe-set invariance near EZs and asymptotic target interception when away from the EZs or maximum engagement range, all while respecting input bounds. Numerical simulations with multiple defenders validated the approach across diverse initial conditions, including challenging configurations with overlapping zones and concave notches. The results show that the proposed EZ-aware safety formulation enables less conservative maneuvering compared with traditional range-based safety constraints, which allows the attacker to exploit favorable engagement geometries while maintaining provable safety and achieving reduced interception times. 

\bibliography{sample}

\end{document}